\documentclass[11pt,letterpaper]{amsart}

\usepackage{header}
\numberwithin{equation}{subsection}
\title[TBG in Commensurate Angles]{Twisted Bilayer Graphene in Commensurate Angles}
\author[Tal Malinovitch]{Tal Malinovitch$^1$\\
\tiny $^1$Department of Mathematics, Rice University, Houston, TX, 77005}

\begin{document}
\begin{abstract}
    We study a 2D continuum model of electronic transport in twisted bilayer graphene (TBG) at commensurate angles. We use two honeycomb potentials with the symmetries of graphene, either sharing a common origin (AA stacking) or shifted by a half-lattice spacing (AB stacking), and twisted relative to each other. While the electronic properties of TBG are most commonly studied via the approximate Bistritzer-MacDonald (BM) model, our approach studies the exact continuum Schr\"{o}dinger operator without these approximations. Our results hold for a wide class of potentials in both stacking types. We describe the exact angles for which the two twisted lattices are commensurate and prove the existence of Dirac cones at the vertices of the Brillouin zone for such angles. Additionally, we establish quantitative bounds showing that, for small potentials, the slope of the Dirac cones flattens at commensurate angles near incommensurate angles. This work is the first to rigorously establish the existence of Dirac cones for twisted bilayer graphene in the continuum setting, without the BM approximations.
\end{abstract}

\maketitle

\section{Introduction }
\subsection{Motivation and main results}
The electronic properties of twisted bilayer graphene (which we abbreviate as TBG) have been studied extensively in recent years in the physics community (e.g. \cite{bistritzer2011moire, cao2018unconventional, hennighausen2021twistronics, koshino2019band, ledwith2021strong, li2022Dirac, neto2009electronic, yankowitz2019tuning}), and in the mathematical community (for example, \cite{becker2021spectral, becker2022mathematics, becker2024fine, becker2024dirac, cances2023simple, massatt2021electronic,massatt2018incommensurate,watson2023bistritzer, watson2021existence}).  \par 
This system was famously studied theoretically by Bistritzer and MacDonald in their seminal paper \cite{bistritzer2011moire}- where they considered two layers of graphene, one on top of the other, shifted with respect to each other, and twisted one layer with respect to the other, at some angle $\theta$. In their work, Bistritzer and MacDonald created an effective model for TBG, which is periodic at all twisting angles- which we refer to as the BM model. They derived this model using successive approximations to enable more tractable analysis.  One of the main predictions was that the resulting operator's spectrum should contain a degenerate Dirac cone (more on Dirac cones in Section \ref{Setting}) at certain angles, called magic angles. Since then, there has been much effort to establish these results rigorously (e.g., \cite{cances2023simple, watson2023bistritzer}), both in this model and in related models (such as the one in \cite{tarnopolsky2019origin}, for example, in works like \cite{becker2021spectral, becker2022mathematics}). Despite these efforts, some aspects of the BM approximations remain poorly understood, especially in the continuum setting.\par 
In this work, we study a continuous model of TBG, without these approximations, at commensurate angles. A \emph{commensurate angle} is an angle for which the twisted system is still periodic (we use a slightly different definition in the following, see remark \ref{Rmk:Commensurate}), thus allowing the use of Floquet theory and providing a natural definition for Dirac cones, without relying on the approximation of the BM model. \par
This work will focus on the continuum setting- though a similar analysis can be carried out for the discrete operators, known as the tight binding model.\par
To state our result more precisely, let $V$ and $Q$ be honeycomb potentials, as defined in \cite{fefferman2012honeycomb} (see Section \ref{Setting}, and remark \ref{Rmk:Ambiguity} for more about that)- a periodic potential with the honeycomb lattice symmetries (periodic with respect to the same lattice). We denote by $R_\theta$ the matrix that represents a rotation by $\theta$, and $\calR_\theta$ its corresponding operator. Then our twisted potential is (for AA stacking):
\begin{align}
    W^\theta_{AA}=G(\calR_\theta V,\calR_{-\theta} Q)
\end{align}
for a suitable family of admissible interaction operators.
\begin{remark}
For AB stacking, the potential requires more symmetry than the AA stacking case, including requiring $Q=V$, due to some ambiguity when shifting an infinite system to AB stacking; see remark \ref{Rmk:Ambiguity} for more details. 
\end{remark}
Then we consider the following Hamiltonian acting on $L^2(\bbR^2)$
\begin{align*}
    H^\theta(\lambda)=-\Delta +\lambda W^\theta_{AA/AB} 
\end{align*}
for $\lambda \in \bbR$, the amplitude of the potential. As a representative example of the class of potentials considered here, one may consider, for the AA stacking  
\begin{align*}
    W^\theta_{0,AA} =\frac{1}{2}(V(R_{\theta}x)+Q(R_{-\theta}x))
\end{align*}

\begin{remark}\label{3dRmk}
    We note that the model proposed here is a two-dimensional scalar model. First, regarding the spatial dimensions, the model here is similar to the celebrated BM model and its rigorous derivations \cite{bistritzer2011moire, cances2023simple, watson2023bistritzer}, which are essentially two-dimensional. 
    The third dimension is incorporated into the BM model via a weak interlayer coupling, and in our framework, this interaction is captured by the choice of the continuous operator $G$. We provide a concrete demonstration of this out-of-plane slicing in Example \ref{3dEx} below.
    Second, regarding the scalar nature of the Hamiltonian: while the BM model utilizes a matrix structure to represent the bilayer system, this is an emergent artifact of their approximations. Specifically, the BM Hamiltonian arises from discretizing the continuous $z$-axis to force electrons into discrete layer states (a tight-binding approximation) and truncating momentum space to isolate the Dirac cones (a valley restriction). In reality, the fundamental physics of an electron moving through a TBG potential at non-relativistic energies is governed by a continuous scalar operator. By studying the exact scalar potential $H = -\Delta + \lambda W$ in the continuous setting, our model avoids the artificial distinctions introduced by these tight-binding and valley restrictions. This provides a rigorous mathematical foundation to establish the existence of Dirac cones directly from the continuum, rather than projecting the continuous reality onto approximate discrete bases.
\end{remark}
With this model, we may state our main results (more precise statements will appear in Section \ref{Setting}- after some more technical notations will be introduced):
\begin{enumerate}
    \item Theorem \ref{NewLattice} describes the set of commensurate angles- the set of angles for which the two lattices intersect non-trivially, denoted by $\calC$. Furthermore, we show that, for a commensurate angle, the new potential is periodic with respect to a scaled honeycomb lattice with a scaling factor of $N$, defined by the arithmetic properties of the angle $\theta$. 
    \item Theorem \ref{DiracPoint} proves that for such systems, for every $\lambda \in \bbR$ except for a discrete set, under some technical conditions, there are Dirac points at the edges of the new Brillouin zone, near the bottom of the spectrum. This relies on the results and techniques of \cite{fefferman2012honeycomb} and \cite{berkolaiko2018symmetry}. Though it may seem that this result follows directly from these papers, for some cases of admissible interacting potentials (such as the example above), the technical condition appearing in \cite{fefferman2012honeycomb, berkolaiko2018symmetry} (non-vanishing of some Fourier coefficient of the potential) does not hold. This condition arises from a perturbative argument, and we show that, by going to higher orders, we can still recover their result under a different condition, thus allowing us to find a condition that does not fail for all admissible potentials of this form.  
    \item  As an additional result, Theorem \ref{Vanishing} shows that for a small amplitude of the potential with respect to the reciprocal of the scaling, that is $\lambda \lesssim \frac{1}{N^2}$, the slope of the Dirac cone $v_d$ is proportional to $\frac{1}{N}$. This result may hint at vanishing Floquet bands for \emph{all} incommensurate angles for small enough $\lambda$ and gives a quantitative flattening of the Dirac cones, albeit only in the perturbative regime, for commensurate approximations to incommensurate angles. This result is not of the "magic angle" type, as it is uniform across all incommensurate angles rather than confined to a set of specific angles. 
\end{enumerate}
\subsection{Graphene and twisted bilayer graphene - overview}
\subsubsection{Single layer graphene}
Graphene is a two-dimensional material made of a single layer of carbon atoms arranged in a hexagonal formation. There are several ways of studying such materials. One such way is to examine the associated Schr\"{o}dinger operator in the continuum setting, i.e., as an operator acting on $L^2(\bbR^2)$. Another way is to study these operators through the tight-binding approximation. In this approximation, the full dynamics is approximated by an operator that acts on  $\ell^2(\Lambda)$, where $\Lambda$ is the graph of the periodic lattice of atoms. For a discussion of the tight-binding model, see, for example, \cite{ablowitz2012tight, fefferman2018honeycomb}.\par
For periodic potentials, Floquet theory allows one to move from the spectrum of the full Hamiltonian, $H$, which has absolutely continuous spectrum, to the study of a family of operators $H(k)$, each with only pure point spectrum \cite{Kuchment}. The eigenfunctions of each $H(k)$, denoted $E_n(k)$, are called bands.\par
One remarkable property of graphene is that it has Dirac points. Dirac points are points where two bands - two different eigenvalues, $E_1(k)$ and $E_2(k)$, of $H(k)$ - touch conically. In other words, we say that $(E_0,k_0)$ is a Dirac point in the energy-quasimomentum plane if there is some $\delta>0$, such that for all $k\in \bbT^*$ such that $|k-k_0|<\delta$, we have
\begin{align*}
    &|E_1(k) -E_0|\approx |v_d||k-k_0|\text{ and } &|E_2(k) -E_0|\approx -|v_d||k-k_0|,
\end{align*}
for some $v_d$ - called the Dirac velocity, see Section \ref{Setting} for more precise definition. This means that a wave packet localized in momentum space around that point disperses approximately according to a two-dimensional Dirac equation, the equation of evolution for massless relativistic fermions (see \cite{fefferman2014wave} for more details about the dispersion near Dirac points)- and hence the name.\par
The existence of these Dirac points was shown first in a tight-binding setting in the physics literature in \cite{slonczewski1958band, Wallace}, and in a richer model that was considered in the mathematics literature in \cite{kuchment1973spectra}. Later, it was proven for the continuous setting in the seminal work of Fefferman and Weinstein \cite{fefferman2012honeycomb}- which the present work draws inspiration from. \par
Fefferman and Weinstein modeled a single-layer graphene by a Schr\"{o}dinger operator with a honeycomb potential acting on $L^2(\bbR^2)$. A honeycomb lattice, defined here in Section \ref{Setting}, is, roughly speaking, a potential with the same symmetries as graphene. In \cite{fefferman2012honeycomb}, they showed that this model has, under some mild assumptions, Dirac points at the vertices of the Brillouin zone. Finally, they have shown that these points persist under a broad class of perturbations. \par
Later, in \cite{berkolaiko2018symmetry}, Berkolaiko and Comech gave a different proof of the results in \cite{fefferman2012honeycomb}, which made the role of symmetry in the arguments of \cite{fefferman2012honeycomb} more evident by using more abstract arguments based on representation theory. Thus, they could generalize the results to many more applications and simplify some of the more technical aspects of that work. 

\subsubsection{Twisted bilayer graphene}
As mentioned above, the celebrated model of twisted bilayer graphene was conceived in 2011 by  Bistritzer and MacDonald \cite{bistritzer2011moire}. 
To make the analysis tractable, Bistritzer and MacDonald restricted their attention to quasimomentum near the Dirac cones of the single-layer model. They could approximate the evolution using a periodic operator, regardless of whether the angle is commensurate or incommensurate.\par 
The breakthrough work of Becker, Embree,  Wittsten, and  Zworski \cite{becker2022mathematics}, as well as an alternate proof given by Watson and Luskin \cite{watson2021existence},  showed that some approximation of the BM Hamiltonian contains flat bands- which are tied to the existence of the aforementioned "magic angles". Then, in a series of papers, Becker and collaborators \cite{becker2021spectral,becker2024fine,becker2024dirac2} studied this model in greater detail, providing even more detailed spectral descriptions of the flat bands. In a very recent work, Becker, Quinn, Tao, Watson, and Yang \cite{becker2024dirac} established the existence of Dirac cones and the existence of magic angles with degenerate Dirac cones for the full BM model.  \par
In the last couple of years, there has been an attempt to better understand the approximations leading to the BM Hamiltonian by Canc{\`e}s, Garrigue, and Gontier \cite{cances2023simple} in the continuum setting, and by Watson, Kong, Macdonald, and Luskin \cite{watson2023bistritzer} in the tight-binding setting. These studies have rigorously estimated the error terms from the derivation of the BM Hamiltonian from the original Dirac equation. Thus, their results bound the error when comparing the evolution of the full operator to that of the approximate operator. This bound is time-dependent. For a recent survey of the results in this field, see \cite{zworski2024mathematical}.\par
There are still aspects of the BM model that are not well understood, especially the approximation in the continuum setting (for example, defining the Kohn-Sham potential for incommensurate systems, see \cite{cances2023simple}). Moreover, superconductivity usually arises from some spectral phenomenon in the single-particle theory. Generally speaking, approximations to the evolution on certain time scales, such as the BM model to the full model, do not provide information about the spectral properties. For that, a different notion of convergence is usually required. So, we would like to build a more fundamental understanding of the electronic transport properties of TBG without using this model's assumptions. Specifically, we establish the existence of Dirac cones for commensurate angles without going through the approximate BM model, but rather directly from a more general description. 

\subsection{Outline of the paper}
This paper is organized as follows: \par
Section \ref{Setting} introduces the basic setting and notation, as well as the main tools of Floquet theory, which enable us to state our main results precisely. \par 
Section \ref{SecondOrderPert} will prove Theorem \ref{DiracPoint}, which allows us to conclude the existence of Dirac points for honeycomb potentials with slightly different conditions than the main theorems in \cite{berkolaiko2018symmetry,fefferman2012honeycomb}- thus enabling us to use them for a larger family of twisted bilayer potentials.\par
Section \ref{TBP} will prove our main results regarding the twisted bilayer potentials with commensurate angles. First, we prove Theorem \ref{NewLattice}, which shows that our potential is periodic with respect to a scaling of a honeycomb potential. Then, we show that for some examples of twisted potential, the condition in \cite{berkolaiko2018symmetry, fefferman2012honeycomb} does not hold, and provide a description of our new technical condition- Lemma \ref{WDescription}. Finally, we show that for a small enough coupling constant, the Dirac velocity decays like the reciprocal of the scaling factor in Theorem \ref{Vanishing}.  \par
Section \ref{Example} will give examples of a twisted potential such that for all angles, the technical condition of Theorem \ref{DiracPoint} holds.\par
Finally, Appendix \ref{ExpLatCom} contains some explicit lattice computations, and Appendix \ref{Notations} collects the notation used in this paper.

\subsection*{Acknowledgment} The author thanks Adam Black, Long Li, Giorgio Young, and Elad Zelingher for many discussions on this problem. In addition, the author would also like to thank Mitchell Luskin, Alex Watson, and their group, as well as Svetlana Jitomirskaya, Matthew Powell, and Peter Kuchment for fruitful discussions.

\section{The setting and results}\label{Setting}
\subsection{Geometry}
We start with the definitions of honeycomb potentials and lattices. We mostly follow the notations conventions set in \cite{fefferman2012honeycomb}. We recall the honeycomb lattice\footnote{This is, in fact, a triangular lattice, but it turns out that for this analysis, this is enough- see discussion in \cite{fefferman2012honeycomb}. } is given by: 
\begin{align*}
    &v_1=\begin{pmatrix}\frac{\sqrt{3}}{2}\\\frac{1}{2}\end{pmatrix}, v_2=\begin{pmatrix}\frac{\sqrt{3}}{2}\\-\frac{1}{2}\end{pmatrix},&\Lambda=v_1\bbZ\oplus v_2\bbZ
\end{align*}
We also consider the reciprocal lattice, defined by 
 \begin{align*}
      & k_1=\frac{4\pi}{\sqrt{3}}\begin{pmatrix} \frac{1}{2}\\\frac{\sqrt{3}}{2}\end{pmatrix},k_2=\frac{4\pi}{\sqrt{3}}\begin{pmatrix} \frac{1}{2}\\-\frac{\sqrt{3}}{2}\end{pmatrix},&\Lambda^*=k_1\bbZ\oplus k_2\bbZ
\end{align*}
It is convenient to define the following matrices
\begin{align*}
    &\nu =\begin{pmatrix}v_1 &v_2\end{pmatrix}=\begin{pmatrix}\frac{\sqrt{3}}{2}&\frac{\sqrt{3}}{2}\\\frac{1}{2}&-\frac{1}{2} \end{pmatrix},&\kappa=\begin{pmatrix}k_1 &k_2\end{pmatrix}=\frac{4\pi}{\sqrt{3}}\begin{pmatrix} \frac{1}{2}&\frac{1}{2}\\\frac{\sqrt{3}}{2}&-\frac{\sqrt{3}}{2} \end{pmatrix}
\end{align*}
Both lattices can be written in the following form:
\begin{align*}
    \Lambda=\nu \bbZ^2, \Lambda^*=\kappa \bbZ^2
\end{align*}
and we note that for any $u_1,u_2\in \bbZ^2$ we have that $\braket{\kappa u_1,\nu u_2}=2\pi \braket{u_1,u_2}$.\par
We distinguish between the Euclidean inner product, which we denote by $\braket{\cdot,\cdot}$, and the inner product on Hilbert spaces, which we denote by $(\cdot,\cdot)$, for clarity. \par
Throughout this paper, quantities with a tilde above them, such as $\tilde{\nu} $, denote quantities related to a honeycomb lattice, $\tilde{\Lambda}$, without explicit dependence on its base vectors. $\Lambda$ and $\Lambda^*$ always refer to the above choices of base vectors, and quantities with the superscript of $\theta$ refer to quantities related to the new lattice generated by the intersection of twisted lattices by commensurate angle $\theta$. \par
For any honeycomb lattice, $\tilde{\Lambda}$, with base matrix $\tilde{\nu}$, and dual matrix $\tilde{\kappa}$, we may define the unit cell $\tilde{\Omega}$, and the Brillouin zone, $\tilde{\calB}$ by
\begin{align*}
    &\tilde{\Omega}=\tilde{\nu}[0,1]^2,&\tilde{\calB}=\{k\in \bbR^2\mid \forall a\in \tilde{\Lambda}^*, |k|\leq |k-a|\}
\end{align*}
Next, we denote the rotation matrix by angle $\theta$ as  
\begin{align*}
    &R_\theta=\begin{pmatrix}\cos(\theta)& \sin(\theta)\\-\sin(\theta)& \cos(\theta)\end{pmatrix}
\end{align*}
and we denote the corresponding operator by $\calR_{\theta}$, that is $\calR_{\theta} f(x)=f(R_{-\theta}x)$.\par
We single out the rotation by $\frac{2\pi}{3}$, by denoting $R=R_{\frac{2}{3}\pi}$, and the corresponding operator we denote by $\calR$. \par
The points of high symmetry in the Brillouin zone are of particular importance- these are points where rotation by $R$ results in a shift by the dual lattice:
\begin{align*}
    \tilde{\bbP}=\{\vec{k}\in \tilde{\calB}\mid (R-\id)\vec{k}\in \tilde{\kappa} \bbZ^2\}
\end{align*}
Moreover, we can decompose it into three disjoint orbits
\begin{align*}
    \tilde{\bbP}=\{\tilde{K},R\tilde{K},R^2\tilde{K}\}\bigsqcup \{\tilde{K}',R\tilde{K}',R^2\tilde{K}'\}\bigsqcup\{0\}
\end{align*}
with $\tilde{K}=\frac{1}{3}\begin{cases}  \kappa \begin{pmatrix}1\\-1\end{pmatrix}, &\tilde{\Lambda}=\Lambda \\ \nu \begin{pmatrix}1\\1\end{pmatrix},&\tilde{\Lambda}=\Lambda^* \end{cases}$ and $\tilde{K}'=-\tilde{K}$.
We consider the shift operators at the appropriate high-symmetry points. For $\vec{a}\in \bbR^2$, we may consider the translation operator $\calT_{\vec{a}}$ given by $\calT_{\vec{a}}f(x)=f(x-\vec{a})$.\par
Specifically, we are interested in the shift by $K_0=\frac{1}{3}\nu \begin{pmatrix}1\\1\end{pmatrix}=\begin{pmatrix}\frac{1}{\sqrt{3}}\\0\end{pmatrix}$, and its dual point $K_0^*=\frac{1}{3}\kappa \begin{pmatrix}1\\-1\end{pmatrix}=\frac{4\pi}{\sqrt{3}}\begin{pmatrix}0\\\sqrt{3}\end{pmatrix}$, and their associated shift operators $\calT_{K_0}$, and $\calT_{K_0^*}$. These shifts allow us to define the AB stacking potentials\footnote{We remark that one can write the hexagonal lattice of graphene as $(\Lambda-\frac{1}{2}K_0) \cup (\Lambda+\frac{1}{2}K_0)$.}.\par
We will also use the parity operator $Ff(x)=f(-x)$, which is equivalent to $\calR_{\pi}$. 
\begin{remark}
    Naturally, one may also consider the opposite shift, which leads to BA stacking. All the analysis below is also valid for this configuration. 
\end{remark}
With this in hand, we recall the definition of a honeycomb potential given in \cite{fefferman2012honeycomb}- extended to treat $\Lambda$ and $\Lambda^* $ on equal footing:
\begin{definition}
    A real-valued potential $U\in C^\infty (\bbR^2)$ is called a honeycomb potential if it satisfies the following properties, with respect to $\tilde{\Lambda}\in \{ \Lambda,\Lambda^*\}$:
    \begin{enumerate}
        \item It is periodic: $\forall a\in \tilde{\Lambda}, x\in \bbR^2, U(x+a)=U(x) $.
        \item It is even: $F[U](x)=U(-x)=U(x)$.
        \item It is symmetric under rotation by $R$, i.e. $\forall x\in \bbR^2,\calR[U](x)=U(R^{-1}x)=U(x)$.
    \end{enumerate} 
\end{definition}
For a list of examples of honeycomb lattices, we refer the reader to \cite{fefferman2012honeycomb}.\par
To define the twisted potential, we define the set of admissible interaction operators
\begin{definition}
     $G: (C^\infty\cap L^\infty) \times (C^\infty\cap L^\infty) \rightarrow C^\infty\cap L^\infty  $ is called an admissible interaction operator if it has the following properties, for any $f,h 
     \in C^\infty \cap L^\infty $
    \begin{enumerate}
        \item Its arguments bound it in the sense that there are some $C_g, C_{g'}>0$ and $\gamma,\gamma'>0$  such that 
        \begin{align*}
            &\|G(f,h)\|_\infty \leq C_g(\|f\|_\infty\|h\|_\infty)^{\gamma}\\
            &\|\nabla G(f,h)\|_\infty \leq C_{g
            }(\|\nabla f\|_\infty  \|\nabla h\|_\infty)^{\gamma'}
        \end{align*}
        \item\label{RotCom} $G$ respects rotations in the following sense
        \begin{align*}
            &\calR_{\alpha}G(f,h)=G(\calR_{\alpha}f,\calR_{\alpha}h)
        \end{align*}
        \item\label{Transcom} $G$ respects translations in the following sense
        \begin{align*}
            &\calT_{\textbf{a}}G(f,h)=G(\calT_{\textbf{a}}f,\calT_{\textbf{a}}h)
        \end{align*}
    \end{enumerate}
    We call $G^*$ an admissible interaction operator for AB stacking if, on top of the above, we have that, it is symmetric, i.e., for any $f,g \in C^\infty \cap L^\infty $: 
    \begin{align*}
        G^*(f,g)=G^*(g,f)
    \end{align*}
\end{definition}
\begin{remark}
    Note that Properties \ref{RotCom} and \ref{Transcom} are quite natural and can be expressed as the fact that spatial position dependence comes from $f,g$ and not from the definition of $G$. 
\end{remark}
With these definitions in hand, we define the twisted bilayer potential of angle $\theta$, which we denote by $W^\theta$, in AA stacking and AB stacking:
\begin{definition}
    Let $V$ and $Q$ be honeycomb potentials with $\Lambda$ as a lattice, and let $G$ be an admissible interaction operator, and $G^*$ an admissible interaction operator for AB stacking, then the corresponding twisted bilayer potential in AA stacking of angle $\theta$ is defined by
    \begin{align*}
        W_{AA}^\theta=G(\calR_\theta V,\calR_{-\theta} Q)
    \end{align*}
    The twisted bilayer potential in AB stacking of angle $\theta$ is defined by
    \begin{align*}
        W_{AB}^\theta=G^*(\frac{1}{3}\sum_{j=-1}^1\calT_{-\frac{1}{2}R^jK_0}\calR_\theta V,\frac{1}{3}\sum_{j=-1}^1\calT_{\frac{1}{2}R^jK_0}\calR_{-\theta} V)
    \end{align*}
\end{definition}
\begin{remark}\label{Rmk:Ambiguity}
    We note that there is some intrinsic ambiguity in AB stacking. Since the shift is taking place before the twist, shifting by $K_0$, $RK_0$, or $R^{-1}K_0$ all result in the same configuration. This is demonstrated in Figure \ref{ABStacking}. So, we average over the three options, ensuring that rotational symmetry is preserved.  
\end{remark}
\begin{figure}[h]
     \centering
     \begin{subfigure}[t]{0.45\textwidth}
         \centering
        \begin{tikzpicture}[scale=1]
        \pgfmathsetmacro {\sqr}{sqrt(3)}
        \pgfmathsetmacro {\vx}{0.5*\sqr}
        \pgfmathsetmacro {\vy}{0.5}
        \pgfmathsetmacro {\shift}{1/\sqr }
        
        \draw[blue, ultra thick](0,0) circle (2pt);
        \draw[blue, ultra thick](\vx,\vy) circle (2pt);
        \draw[blue, ultra thick](\vx,-\vy) circle (2pt);
        \draw[blue, ultra thick](2*\vx,0) circle (2pt);
        \draw[blue, ultra thick](-\vx,-\vy) circle (2pt);
        \draw[blue, ultra thick](-\vx,\vy) circle (2pt);
        \draw[blue, ultra thick](2*\vx,0) circle (2pt);
        \filldraw[blue, ultra thick](0+\shift,0) circle (2pt);
        \filldraw[blue, ultra thick](\vx+\shift,\vy) circle (2pt);
        \filldraw[blue, ultra thick](\vx+\shift,-\vy) circle (2pt) ;
        
        \filldraw[blue, ultra thick](-\vx+\shift,-\vy) circle (2pt);
        \filldraw[blue, ultra thick](-\vx+\shift,\vy) circle (2pt) ;
        \filldraw[blue, ultra thick](-2*\vx+\shift,0) circle (2pt) ;
        
        \draw[blue, thick] (0+\shift,0) --(\vx,\vy) --(\vx+\shift,\vy) --(2*\vx,0) --(\vx+\shift,-\vy) --(\vx,-\vy)--(0+\shift,0) -- cycle;
        \draw[blue, thick] (0,0) --(-\vx+\shift,\vy) --(-\vx,\vy) --(-2*\vx+\shift,0) -- (-\vx,-\vy) --(-\vx+\shift,-\vy)--(0,0) -- cycle;
        \draw[blue, thick] (0,0) --(0+\shift,0)--(0,0) -- cycle;

        \draw[red, ultra thick](0+\shift,0) circle (2pt);
        \draw[red, ultra thick](\vx+\shift,\vy) circle (2pt);
        \draw[red, ultra thick](\vx+\shift,-\vy) circle (2pt);
        \draw[red, ultra thick](2*\vx+\shift,0) circle (2pt);
        \draw[red, ultra thick](-\vx+\shift,-\vy) circle (2pt);
        \draw[red, ultra thick](-\vx+\shift,\vy) circle (2pt);
        \draw[red, ultra thick](2*\vx+\shift,0) circle (2pt);
        \filldraw[red, ultra thick](0+\shift+\shift,0) circle (2pt);
        \filldraw[red, ultra thick](\vx+\shift+\shift,\vy) circle (2pt);
        \filldraw[red, ultra thick](\vx+\shift+\shift,-\vy) circle (2pt) ;
        
        \filldraw[red, ultra thick](-\vx+\shift+\shift,-\vy) circle (2pt);
        \filldraw[red, ultra thick](-\vx+\shift+\shift,\vy) circle (2pt) ;
        \filldraw[red, ultra thick](-2*\vx+\shift+\shift,0) circle (2pt) ;
        
        \draw[red, thick] (0+\shift+\shift,0) --(\vx+\shift,\vy) --(\vx+\shift+\shift,\vy) --(2*\vx+\shift,0) --(\vx+\shift+\shift,-\vy) --(\vx+\shift,-\vy)--(0+\shift+\shift,0) -- cycle;
        \draw[red, thick] (0+\shift,0) --(-\vx+\shift+\shift,\vy) --(-\vx+\shift,\vy) --(-2*\vx+\shift+\shift,0) -- (-\vx+\shift,-\vy) --(-\vx+\shift+\shift,-\vy)--(0+\shift,0) -- cycle;
        \draw[red, thick] (0+\shift,0) --(0+\shift+\shift,0)--(0+\shift,0) -- cycle;
    \end{tikzpicture}
    \caption{Illustration of AB stacking with a shift by $K_0$.}
    \label{ABStackK0}
     \end{subfigure}
     \hfill
     \begin{subfigure}[t]{0.45\textwidth}
         \centering
        \begin{tikzpicture}[scale=1]
        \pgfmathsetmacro {\sqr}{sqrt(3)}
        \pgfmathsetmacro {\vx}{0.5*\sqr}
        \pgfmathsetmacro {\vy}{0.5}
        \pgfmathsetmacro {\shift}{1/\sqr }
        \pgfmathsetmacro {\shiftRKx}{-1/(2*\sqr) }
        \pgfmathsetmacro {\shiftRKy}{-1/(2) }
        
        \draw[blue, ultra thick](0,0) circle (2pt);
        \draw[blue, ultra thick](\vx,\vy) circle (2pt);
        \draw[blue, ultra thick](\vx,-\vy) circle (2pt);
        \draw[blue, ultra thick](2*\vx,0) circle (2pt);
        \draw[blue, ultra thick](-\vx,-\vy) circle (2pt);
        \draw[blue, ultra thick](-\vx,\vy) circle (2pt);
        \draw[blue, ultra thick](2*\vx,0) circle (2pt);
        \filldraw[blue, ultra thick](0+\shift,0) circle (2pt);
        \filldraw[blue, ultra thick](\vx+\shift,\vy) circle (2pt);
        \filldraw[blue, ultra thick](\vx+\shift,-\vy) circle (2pt) ;
        
        \filldraw[blue, ultra thick](-\vx+\shift,-\vy) circle (2pt);
        \filldraw[blue, ultra thick](-\vx+\shift,\vy) circle (2pt) ;
        \filldraw[blue, ultra thick](-2*\vx+\shift,0) circle (2pt) ;
        
        \draw[blue, thick] (0+\shift,0) --(\vx,\vy) --(\vx+\shift,\vy) --(2*\vx,0) --(\vx+\shift,-\vy) --(\vx,-\vy)--(0+\shift,0) -- cycle;
        \draw[blue, thick] (0,0) --(-\vx+\shift,\vy) --(-\vx,\vy) --(-2*\vx+\shift,0) -- (-\vx,-\vy) --(-\vx+\shift,-\vy)--(0,0) -- cycle;
        \draw[blue, thick] (0,0) --(0+\shift,0)--(0,0) -- cycle;

        \draw[red, ultra thick](0+\shiftRKx,0+\shiftRKy) circle (2pt);
        \draw[red, ultra thick](\vx+\shiftRKx,\vy+\shiftRKy) circle (2pt);
        \draw[red, ultra thick](\vx+\shiftRKx,-\vy+\shiftRKy) circle (2pt);
        \draw[red, ultra thick](2*\vx+\shiftRKx,0+\shiftRKy) circle (2pt);
        \draw[red, ultra thick](-\vx+\shiftRKx,-\vy+\shiftRKy) circle (2pt);
        \draw[red, ultra thick](-\vx+\shiftRKx,\vy+\shiftRKy) circle (2pt);
        \draw[red, ultra thick](2*\vx+\shiftRKx,0+\shiftRKy) circle (2pt);
        \filldraw[red, ultra thick](0+\shift+\shiftRKx,0+\shiftRKy) circle (2pt);
        \filldraw[red, ultra thick](\vx+\shift+\shiftRKx,\vy+\shiftRKy) circle (2pt);
        \filldraw[red, ultra thick](\vx+\shift+\shiftRKx,-\vy+\shiftRKy) circle (2pt) ;
        
        \filldraw[red, ultra thick](-\vx+\shift+\shiftRKx,-\vy+\shiftRKy) circle (2pt);
        \filldraw[red, ultra thick](-\vx+\shift+\shiftRKx,\vy+\shiftRKy) circle (2pt) ;
        \filldraw[red, ultra thick](-2*\vx+\shift+\shiftRKx,0+\shiftRKy) circle (2pt) ;
        
        \draw[red, thick] (0+\shift+\shiftRKx,0+\shiftRKy) --(\vx+\shiftRKx,\vy+\shiftRKy) --(\vx+\shift+\shiftRKx,\vy+\shiftRKy) --(2*\vx+\shiftRKx,0+\shiftRKy) --(\vx+\shift+\shiftRKx,-\vy+\shiftRKy) --(\vx+\shiftRKx,-\vy+\shiftRKy)--(0+\shift+\shiftRKx,0+\shiftRKy) -- cycle;
        \draw[red, thick] (0+\shiftRKx,0+\shiftRKy) --(-\vx+\shift+\shiftRKx,\vy+\shiftRKy) --(-\vx+\shiftRKx,\vy+\shiftRKy) --(-2*\vx+\shift+\shiftRKx,0+\shiftRKy) -- (-\vx+\shiftRKx,-\vy+\shiftRKy) --(-\vx+\shift+\shiftRKx,-\vy+\shiftRKy)--(0+\shiftRKx,0+\shiftRKy) -- cycle;
        \draw[red, thick] (0+\shiftRKx,0+\shiftRKy) --(0+\shift+\shiftRKx,0+\shiftRKy)--(0+\shiftRKx,0+\shiftRKy) -- cycle;
    \end{tikzpicture}
    \caption{Illustration of AB stacking with a shift by $RK_0$.}
    \label{ABStackRK0}
     \end{subfigure}
     \hfill
     \begin{subfigure}[t]{0.45\textwidth}
         \centering
        \begin{tikzpicture}[scale=1]
        \pgfmathsetmacro {\sqr}{sqrt(3)}
        \pgfmathsetmacro {\vx}{0.5*\sqr}
        \pgfmathsetmacro {\vy}{0.5}
        \pgfmathsetmacro {\shift}{1/\sqr }
        \pgfmathsetmacro {\shiftRKx}{-1/(2*\sqr) }
        \pgfmathsetmacro {\shiftRKy}{1/(2) }
        
        \draw[blue, ultra thick](0,0) circle (2pt);
        \draw[blue, ultra thick](\vx,\vy) circle (2pt);
        \draw[blue, ultra thick](\vx,-\vy) circle (2pt);
        \draw[blue, ultra thick](2*\vx,0) circle (2pt);
        \draw[blue, ultra thick](-\vx,-\vy) circle (2pt);
        \draw[blue, ultra thick](-\vx,\vy) circle (2pt);
        \draw[blue, ultra thick](2*\vx,0) circle (2pt);
        \filldraw[blue, ultra thick](0+\shift,0) circle (2pt);
        \filldraw[blue, ultra thick](\vx+\shift,\vy) circle (2pt);
        \filldraw[blue, ultra thick](\vx+\shift,-\vy) circle (2pt) ;
        
        \filldraw[blue, ultra thick](-\vx+\shift,-\vy) circle (2pt);
        \filldraw[blue, ultra thick](-\vx+\shift,\vy) circle (2pt) ;
        \filldraw[blue, ultra thick](-2*\vx+\shift,0) circle (2pt) ;
        
        \draw[blue, thick] (0+\shift,0) --(\vx,\vy) --(\vx+\shift,\vy) --(2*\vx,0) --(\vx+\shift,-\vy) --(\vx,-\vy)--(0+\shift,0) -- cycle;
        \draw[blue, thick] (0,0) --(-\vx+\shift,\vy) --(-\vx,\vy) --(-2*\vx+\shift,0) -- (-\vx,-\vy) --(-\vx+\shift,-\vy)--(0,0) -- cycle;
        \draw[blue, thick] (0,0) --(0+\shift,0)--(0,0) -- cycle;

        \draw[red, ultra thick](0+\shiftRKx,0+\shiftRKy) circle (2pt);
        \draw[red, ultra thick](\vx+\shiftRKx,\vy+\shiftRKy) circle (2pt);
        \draw[red, ultra thick](\vx+\shiftRKx,-\vy+\shiftRKy) circle (2pt);
        \draw[red, ultra thick](2*\vx+\shiftRKx,0+\shiftRKy) circle (2pt);
        \draw[red, ultra thick](-\vx+\shiftRKx,-\vy+\shiftRKy) circle (2pt);
        \draw[red, ultra thick](-\vx+\shiftRKx,\vy+\shiftRKy) circle (2pt);
        \draw[red, ultra thick](2*\vx+\shiftRKx,0+\shiftRKy) circle (2pt);
        \filldraw[red, ultra thick](0+\shift+\shiftRKx,0+\shiftRKy) circle (2pt);
        \filldraw[red, ultra thick](\vx+\shift+\shiftRKx,\vy+\shiftRKy) circle (2pt);
        \filldraw[red, ultra thick](\vx+\shift+\shiftRKx,-\vy+\shiftRKy) circle (2pt) ;
        
        \filldraw[red, ultra thick](-\vx+\shift+\shiftRKx,-\vy+\shiftRKy) circle (2pt);
        \filldraw[red, ultra thick](-\vx+\shift+\shiftRKx,\vy+\shiftRKy) circle (2pt) ;
        \filldraw[red, ultra thick](-2*\vx+\shift+\shiftRKx,0+\shiftRKy) circle (2pt) ;
        
        \draw[red, thick] (0+\shift+\shiftRKx,0+\shiftRKy) --(\vx+\shiftRKx,\vy+\shiftRKy) --(\vx+\shift+\shiftRKx,\vy+\shiftRKy) --(2*\vx+\shiftRKx,0+\shiftRKy) --(\vx+\shift+\shiftRKx,-\vy+\shiftRKy) --(\vx+\shiftRKx,-\vy+\shiftRKy)--(0+\shift+\shiftRKx,0+\shiftRKy) -- cycle;
        \draw[red, thick] (0+\shiftRKx,0+\shiftRKy) --(-\vx+\shift+\shiftRKx,\vy+\shiftRKy) --(-\vx+\shiftRKx,\vy+\shiftRKy) --(-2*\vx+\shift+\shiftRKx,0+\shiftRKy) -- (-\vx+\shiftRKx,-\vy+\shiftRKy) --(-\vx+\shift+\shiftRKx,-\vy+\shiftRKy)--(0+\shiftRKx,0+\shiftRKy) -- cycle;
        \draw[red, thick] (0+\shiftRKx,0+\shiftRKy) --(0+\shift+\shiftRKx,0+\shiftRKy)--(0+\shiftRKx,0+\shiftRKy) -- cycle;
    \end{tikzpicture}
    \caption{Illustration of AB stacking with a shift by $R^2K_0$.}
    \label{ABStackRRK0}
     \end{subfigure}
     \caption{Illustration of AB stacking of graphene with the different possible shifts. An illustration of shift by $K_0$ in \ref{ABStackK0}, shift by $RK_0$ in \ref{ABStackRK0}, and a shift by $R^2K_0$ in \ref{ABStackRRK0}. In all cases, the red hexagon corresponds to the upper layer, and the blue corresponds to the lower layer.}
     \label{ABStacking}
\end{figure}
\begin{remark}
    Note that for AB stacking, we require both copies to be of the same potential. Otherwise, this will break the flip symmetry (as the layers are exchanged).  
\end{remark}
As shown later, for AB stacking at commensurate angles, $W^\theta_{AB}$ is not a standard honeycomb potential but rather an almost honeycomb potential. We define this as follows:
\begin{definition}
    A real-valued potential $U^\theta \in C^\infty(\bbR^2)$ is called an almost honeycomb potential if it satisfies the following properties with respect to $\tilde{\Lambda}\in \{\Lambda, \Lambda^*\}$
    \begin{enumerate}
        \item It is periodic: $\forall a\in \tilde{\Lambda}, x\in \bbR^2, U^\theta(x+a)=U^\theta(x) $.
        \item It is invariant under flips: $F^*[U](x)=U^{-\theta}(-x)=U^\theta(x)$.
        \item It is symmetric under rotation by $R$, i.e. $\forall x\in \bbR^2,\calR[U^\theta](x)=U^\theta(R^{-1}x)=U^\theta(x)$.
    \end{enumerate}
\end{definition}
\begin{examples}\label{Examples}
    For the convenience of the reader, we collect some examples of admissible interaction operators, all of which can be easily verified to have the properties above: 
    \begin{enumerate}
        \item  As mentioned in the introduction, one may take a simple model for $G$, as a concrete example:
        \begin{align*}
            G(f,g)=\frac{1}{2}(f+g)
        \end{align*}
        Which results in the following:
        \begin{align}
            &W^\theta_{0,AA} =\frac{1}{2}(V(R_{\theta}x)+Q(R_{-\theta}x))\label{AddPotAA}\\
            &W^\theta_{0,AB} =\frac{1}{6}\sum_{j=-1}^1(V(R_{\theta}(x+\frac{1}{2}R_{\frac{2\pi}{3}}^jK_0))+V(R_{-\theta}(x-\frac{1}{2}R_{\frac{2\pi}{3}}^jK_0)))\label{AddPotAB}
        \end{align}
        \item One can also take an averaging-type operator: for some  $w\in L^1(\bbR^2)$, we may define 
        \begin{align*}
            G_w(f,h)(x)=f\ast_w h(x)=\int\limits_{\bbR^2} f(z)h(z)w(\|z-x\|)\,dz
        \end{align*}
        \item More generally, one may choose $w\in L^1(\bbR^2)$, and let $p(x,y)$ be a symmetric polynomial in two variables, with actions $(+,\cdot, \ast_w)$ and real coefficients. Then, we may take
        \begin{align*} 
            G(f,h)=p(f,h)
        \end{align*}
        \item \label{3dEx} Consider $U(x,y,z)$ to be a 3D graphene potential (modeling a single layer of graphene in a vacuum), then we can model the effective potential of a 2D slice at height $z_0$ between two layers located at $z_{top}$ and $z_{bottom}$ by taking $V(x,y)=U(x,y,z_0-z_{top})$ and $Q(x,y)=U(x,y,z_0-z_{bottom})$. We may then choose $G$ to accommodate any necessary inter-layer interaction corrections, provided they are compatible with the above constraints.
    \end{enumerate}
\end{examples}
\begin{remark}
    The model proposed above is an idealization of TBG in two primary respects. First, while our choice of $V$, $Q$, and $G$ allows us to model out-of-plane coupling (as discussed in Remark \ref{3dRmk} and Example 4 above), the framework reduces a complex three-dimensional system to an effective two-dimensional one. In addition, in an actual system, mechanical relaxation effects change the stacking type (AA, AB, and BA) over the sample. 
\end{remark}

Our first result concerns describing the set of commensurate angles. For this, we start by denoting by $\calC$ the set of $\theta$ for which $\Lambda^\theta=R_{\theta}\Lambda \cap R_{-\theta}\Lambda$ has a non trivial element. We first note that this is enough to get that $\Lambda^\theta$ contains a lattice:
\begin{proposition}
    If $0\neq \mathbf{a}\in \Lambda^\theta$, then we have $\Lambda^\theta$ contains a non-degenerate lattice.
\end{proposition}
\begin{proof}
    We note that if $\textbf{a}\in \Lambda^\theta=R_{\theta}\Lambda \cap R_{-\theta}\Lambda$, then we have that 
    \begin{align*}
        R\textbf{a}\in RR_{\theta}\Lambda \cap RR_{-\theta}\Lambda=R_{\theta}(R\Lambda) \cap R_{-\theta}(R\Lambda)=\Lambda^\theta
    \end{align*}
    Since $\textbf{a}\neq 0 $, we have that $R\textbf{a},\textbf{a}$ are two linearly independent vectors, and so they generate a non-degenerate lattice. 
    And naturally, we have $\forall c\in \bbZ, c\textbf{a},cR\textbf{a} \in \Lambda^\theta$.
    So we conclude that $\Lambda^\theta$ contains a non-degenerate lattice- as needed. 
\end{proof}
We note that, in general, we have that  $\calT_{\textbf{a}}\calR_{\alpha}=\calR_{\alpha}\calT_{R_{-\alpha}\textbf{a}}$. And so we have the following proposition: 
\begin{proposition}\label{Periodicity}
    We have that $W_{x}^\theta$, for $x\in \{AA,AB\}$ is periodic with respect to $\Lambda^\theta$.
\end{proposition} 
\begin{proof}
    We note that for any $\textbf{a}\in \Lambda^\theta$, we have that 
    \begin{align*}
        &\calT_{\textbf{a}}W_{AA}^\theta= G(\calT_{\textbf{a}}\calR_{\theta}V, \calT_{\textbf{a}}\calR_{-\theta}Q) =G(\calR_{\theta}\calT_{R_{-\theta}\textbf{a}}V, \calR_{-\theta}\calT_{R_{\theta}\textbf{a}}Q) =G(\calR_{\theta}V,\calR_{-\theta}Q)=W_{AA}^\theta
    \end{align*}
    since $R_\theta\textbf{a},R_{-\theta}\textbf{a}\in \Lambda$.\par
    For AB stacking we note that for any vector $\textbf{b}\in \bbR^2 $ we have that 
    \begin{align*}
        &G^*(\calT_{\textbf{a}}\calT_{-\textbf{b}}\calR_{\theta}V ,\calT_{\textbf{a}}\calT_{\textbf{b}}\calR_{-\theta}V)  =G^*(\calT_{-\textbf{b}}\calT_{\textbf{a}}\calR_{\theta}V,\calT_{\textbf{b}}\calT_{\textbf{a}}\calR_{-\theta}V)\\
        &=G^*(\calT_{-\textbf{b}}\calR_{\theta}\calT_{R_{-\theta}\textbf{a}}V,\calT_{\textbf{b}}\calR_{-\theta}\calT_{R_{\theta}\textbf{a}}V)=G^*(\calT_{-\textbf{b}}\calR_{\theta}V,\calT_{\textbf{b}}\calR_{-\theta}V)
    \end{align*}
    since translations commute. From the linearity of the translation operator, we get that $\calT_{\textbf{a}}W_{AB}^\theta=W_{AB}^\theta$.   In particular, both potentials are periodic with respect to $\Lambda^\theta$- as claimed. 
\end{proof}
\begin{remark}\label{Rmk:Commensurate}
    One can also define the set of angles that generate commensurate potentials $\tilde{\calC}$- that is, the set of all $\theta $ such that there exists a $0\neq \textbf{a}\in \bbR^2$ such that  $W^\theta (x+\textbf{a})=W^\theta(x)$. It is easy to see that $\calC\subset \tilde{\calC}$.We believe that $\calC= \tilde{\calC}$- though we will not try to prove it here. 
\end{remark}
With this notation, we prove the following 
\begin{theorem}\label{NewLattice}
    We have 
    \begin{align*}
        \theta \in \calC\cap (0,\frac{\pi}{3})\iff \exists 0<b<a, \gcd(b,a)=1, \tan(\theta)=\frac{\sqrt{3}b}{a}
    \end{align*}
    And any other $\tilde{\theta}\in \calC$ can be reduced via the potential symmetries to some $\theta \in \calC\cap [0,\frac{\pi}{3})$. \par
    Furthermore,  if we denote  
    \begin{align*}
        &\alpha =\begin{cases}
            8\pi, & 3\mid a \text{ and } 2\nmid ab\\
            2, & 3\nmid a \text{ and } 2\nmid ab\\
            4\pi , & 3\mid a \text{ and } 2\mid ab\\
           1, & 3\nmid a \text{ and } 2\mid ab
        \end{cases},& N=\frac{1}{\alpha} \sqrt{a^2+3b^2}
    \end{align*}
    Then we have that
    \begin{align*}
        \Lambda^\theta=N\begin{cases}
            \Lambda, & 3\nmid a\\
            \Lambda^* & 3\mid a
        \end{cases}
    \end{align*}
\end{theorem}
\begin{remark}
    Even though the geometry of commensurate angles has been previously considered, see for example \cite{catarina2019twisted,lopes2012continuum, rodriguez2015coincidence,scheer2022magic}, and similar rationality conditions have been derived, to the best of our knowledge, none of the previous results explicitly state that the new lattice is a scaled version of the honeycomb lattice (or the dual of such a lattice).
\end{remark}
Throughout most of this paper, we consider the following operator 
\begin{align}\label{Hamiltonain}
    H^{\theta}(\lambda)=-\Delta +\lambda W^\theta_x
\end{align}
for $\lambda\in \bbR$, $x\in \{AA, AB\}$, and $W^\theta_x$ a twisted bilayer potential of angle $\theta$, for either AA or AB stacking, that corresponds to some honeycomb potential $V$.\par 
An immediate corollary of Theorem \ref{NewLattice} is 
\begin{corollary}
    Let $W^{\theta}_{AA}$ ($W^{\theta}_{AB}$) be a twisted bilayer potential in AA (AB) stacking of angle $\theta$, for $\theta \in \calC\cap (0,\frac{\pi}{3})$, then $W^{\theta}_{AA}$ ($W^{\theta}_{AB}$) is a(n almost) honeycomb potential, with respect to lattice denoted by $\Lambda^\theta\in \{N\Lambda, N\Lambda^*\}$, for $N$ as defined in Theorem \ref{NewLattice}. 
\end{corollary}
\begin{proof}
    By Theorem \ref{NewLattice} and Proposition \ref{Periodicity}, we have that $W_{x}^\theta$ for $x\in \{AA, AB\}$ is periodic with respect to scaled version of $\Lambda$ or $\Lambda^*$. We need to check the symmetries; for AA stacking, we have:
    \begin{align*}
        &\calR W_{AA}^\theta&&= \calR G(\calR_{\theta}V,\calR_{-\theta}Q)=G( \calR_{\theta}\calR V, \calR_{-\theta}\calR Q)=G( \calR_{\theta} V, \calR_{-\theta}Q)=W_{AA}^\theta\\
        &\calR_{\pi} W_{AA}^\theta&&= \calR_{\pi} G(\calR_{\theta}V,\calR_{-\theta}Q)=G( \calR_{\theta}\calR_{\pi} V, \calR_{-\theta}\calR_{\pi} Q)=G( \calR_{\theta} V, \calR_{-\theta}Q)=W_{AA}^\theta
    \end{align*}
    For AB stacking, we have that 
    \begin{align*}
        &\calR W_{AB}^\theta= \calR G^*(\frac{1}{3}\sum_{j=-1}^1\calT_{-\frac{1}{2}R^jK_0}\calR_{\theta}V,\frac{1}{3}\sum_{j=-1}^1\calT_{\frac{1}{2}R^jK_0}\calR_{-\theta}V)\\
        & =G^*( \frac{1}{3}\sum_{j=-1}^1\calT_{-\frac{1}{2}R^{j+1}K_0}\calR_{\theta}\calR V, \frac{1}{3}\sum_{j=-1}^1\calT_{\frac{1}{2}R^{j+1}K_0}\calR_{-\theta}\calR V)\\
        &=W_{AB}^\theta
    \end{align*}
    since we have that $R^2=R^{-1}$.  We also have that 
     \begin{align*}
        &\calR_\pi W_{AB}^{-\theta}= \calR_\pi G^*(\frac{1}{3}\sum_{j=-1}^1\calT_{-\frac{1}{2}R^jK_0}\calR_{-\theta}V,\frac{1}{3}\sum_{j=-1}^1\calT_{\frac{1}{2}R^jK_0}\calR_{\theta}V)\\
        & =G^*( \frac{1}{3}\sum_{j=-1}^1\calT_{\frac{1}{2}R^jK_0}\calR_{-\theta}\calR_\pi V, \frac{1}{3}\sum_{j=-1}^1\calT_{-\frac{1}{2}R^jK_0}\calR_{\theta}\calR _\pi V)\\
        &=W_{AB}^\theta
    \end{align*}
    where we used that $G^*$ is symmetric. 
\end{proof}
\subsection{Floquet theory}
Next, we need to introduce some notions in Floquet's theory for Schr\"{o}dinger operator with honeycomb potentials. This section considers an arbitrary potential $U$, periodic with respect to a honeycomb lattice $\tilde{\Lambda}$, and corresponding unit cell $\tilde{\Omega}$. We consider the operator. 
\begin{align*}
    \tilde{H}=-\Delta+U.
\end{align*}
Define the following spaces
\begin{align*}
    &L^2_{k}(\tilde{\Omega})=\{f\in L^2(\tilde{\Omega} )\mid \forall a\in \tilde{\Lambda}, f(x+a)=e^{-i\braket{k,a}}f(x)\}.
\end{align*}
the spaces of pseudo-periodic functions on the unit cell $\tilde{\Omega}$, for $k\in \tilde{\calB}$. These spaces are equipped with a natural inner product 
\begin{align*}
    \forall f,g \in L^2_k(\tilde{\Omega}), \, (f,g)=\frac{1}{|\tilde{\Omega}|}\int\limits_{\tilde{\Omega}} \bar{f}(x)g(x) \, dx
\end{align*}
where $|\cdot|$ denotes the Lebesgue measure of the set. 
Usually, we suppress unit cell dependence, which should be inferred from the context.  \par 
Define for $f\in L^2(\bbR^{2})$ the \emph{Floquet transform}
\begin{align*}
    &(\calU f)(k,y)=\sum\limits_{\vec{n}\in \bbZ^2} e^{-i\langle k, \nu \vec{n}\rangle }f(y+\nu \vec{n})
\end{align*}
for $y\in\bbR^{2}$ and $k\in \calB$. As an $L^2(\calB)\otimes L^2(\tilde{\Omega})$ convergent sum, the Floquet transform defines a bounded map from $L^2(\bbR^{2})$ to $L^2(\calB)\otimes L^2(\tilde{\Omega}) $.  
 The following properties of the Floquet transform are standard. See, for instance, Sections 4 and 5 of \cite{Kuchment}:
\begin{proposition}\label{floquetPr}
    The map $ f\mapsto \calU f $ has the following properties:
\begin{enumerate}
    \item $\calU$ is a unitary map from $ L^2(\bbR^{2}) $ to $ L^2(\calB)\otimes  L^2(\tilde{\Omega} )$.
    \item We have the unitary equivalence
    \begin{align*}
        \calU\tilde{H}\calU^* =\int\limits_{\calB}^\oplus \tilde{H}(k) \, \frac{dk}{|\calB|}, 
    \end{align*}
    where $\tilde{H}(k)=-\Delta+U,$ acting on $L^2_{k}$ is a self-adjoint operator.
    \item For any $k\in \bbT^*$, $\tilde{H}(k)$ is bounded from below and has only pure point spectrum- so we have 
    \begin{align*}
        E_1(k)\leq E_2(k)\leq \dots 
    \end{align*}
    where $E_n(k) \xrightarrow{n\rightarrow\infty }\infty $. 
\end{enumerate}
\end{proposition}
The reader may find the necessary background on direct integrals of Hilbert spaces in \cite{RSVol4}. \par
We note that for periodic function, i.e., $f\in L^2_{0}=L^2_{per}$, we also have the following Fourier representation
\begin{align*}
    &f(y)=\sum_{\vec{m}\in \bbZ^2}\hat{f}_{\vec{m}} e^{i\braket{\kappa \vec{m},y}}
    &\hat{f}_{\vec{m}}=\frac{1}{|\Omega|}\int\limits_{\Omega}e^{-i\braket{\kappa \vec{m},y}}f(y)\, dy
\end{align*}
This representation will be used mostly in the context of the potential. 
\subsubsection{Rotational symmetry}
On top of the translation symmetry (which allows for the use of the Floquet transform), we also have symmetry with respect to rotation by $R$, as we have that for the honeycomb potential $U$: $\calR [U](x)=U(R^{-1}x)=U(x)$. \par 
Representation theory for $ R$-invariant Hamiltonians allows us to do an isotypic decomposition of the space; see \cite{berkolaiko2018symmetry} for more details. So, we define 
\begin{align*}
    L^2_{k,\sigma}=\{f\in L^2_{k}\mid \calR f=\sigma f\}
\end{align*}
for $\sigma\in \{1,\tau,\bar{\tau}\}$, where $\tau=e^{\frac{2\pi}{3}i}=-\frac{1}{2}+\frac{\sqrt{3}}{2}i$- the cubic root of unity. Moreover, we have that for $\tilde{K}_*\in \tilde{\bbP}$, one of the high symmetry points,  the operator $\tilde{H}(\tilde{K}_*)$ maps  $L^2_{\tilde{K}_*,\sigma}$  to itself- and thus allows us to reduce our study of $\tilde{H}(\tilde{K}_*)$ to its action on each $L^2_{\tilde{K}_*,\sigma}$. \par
It is convenient to introduce the following notation for $\tilde{K}_*\in \tilde{\bbP}$, and $\vec{m}\in \bbZ^2$:
\begin{align*}
    \tilde{K}_*(\vec{m})=\tilde{K}_*+\tilde{\kappa}\vec{m}
\end{align*}
Then, we can define 
\begin{align*}
    &B=\tilde{\kappa}^{-1}R\tilde{\kappa}&&\varrho_1=\tilde{\kappa}^{-1}(R-\id)\tilde{K}_*\\
    &\varrho_{-1}=\tilde{\kappa}^{-1}(R^{-1}-\id)\tilde{K}_*&&\varrho_0=0
\end{align*}
Then, we can write that
\begin{align*}
    &R\tilde{K}_* (\vec{m})=\tilde{K}_* (B\vec{m}+\varrho_{1})&&R^2\tilde{K}_*  (\vec{m})=\tilde{K}_* (B^{-1}\vec{m}+\varrho_{-1}),
\end{align*}
And, similarly to \cite{fefferman2012honeycomb}, we define  $\approx$ to identify the orbit of $\vec{m}$ under $B^j\vec{m}+\varrho_j, j\in \bbZ_3$ (throughout this paper we us $\bbZ_3=\{\pm1,0\}$), and we denote $\calS=\bbZ^2/\approx$.
\begin{remark}
    We would suppress the dependence of $\varrho_{\pm1}$, $B$, and $\calS$, on the exact choice of $\tilde{\kappa}$, which should be inferred from context. 
\end{remark}
We also note that if $U$ is a honeycomb potential, we have that 
\begin{align*}
    \forall \vec{m}\in \bbZ^2, \hat{U}_{B\vec{m}}=\hat{U}_{\vec{m}}
\end{align*}
For the reader's convenience, we provide the explicit forms for $W^\theta$. Recall that there are two cases, $\kappa ^\theta =\frac{1}{N}\kappa$ and $\kappa ^\theta =\frac{1}{N}\nu$:
\begin{align*}
    &B=\begin{cases}\begin{pmatrix} 0& -1\\ 1&-1\end{pmatrix}, &\kappa ^\theta =\frac{1}{N}\kappa\\
    \begin{pmatrix} -1&-1\\1 &0\end{pmatrix},& \kappa ^\theta =\frac{1}{N}\nu
    \end{cases},\varrho_0=0,\varrho_1=\begin{cases}
        \begin{pmatrix}0\\1\end{pmatrix},& \kappa ^\theta =\frac{1}{N}\kappa\\
        \begin{pmatrix}-1\\0\end{pmatrix},& \kappa ^\theta =\frac{1}{N}\nu\\
    \end{cases},\varrho_{-1}=\begin{cases}
        \begin{pmatrix}-1\\0\end{pmatrix},& \kappa ^\theta =\frac{1}{N}\kappa\\
        \begin{pmatrix}0\\-1\end{pmatrix},&\kappa ^\theta =\frac{1}{N}\nu\\
    \end{cases}
\end{align*}
The above assumes $\tilde{K}_*=K$; for $\tilde{K}_*=K'$, one should take $\varrho_j'=-\varrho_j$ for $j\in \bbZ_3$.
\subsection{Main theorems}
To better understand the statement of our main theorem, we recall the main theorems from \cite{berkolaiko2018symmetry, fefferman2012honeycomb} regarding the existence of the Dirac cones, which can be written as: 
\begin{theorem}[\cite{berkolaiko2018symmetry} -Theorems 2.4-2.5, \cite{fefferman2012honeycomb} -Theorem 5.1]\label{FW}
    Let $H=-\Delta +\lambda U$, for $\lambda\in \bbR$ and $U$ a honeycomb potential with $\tilde{\Lambda}=\Lambda$, be such that 
    \begin{align}\label{WFCond}
        \hat{U}_{-\varrho_{-1}}=\frac{1}{|\tilde{\Omega}|}\int\limits_{\tilde{\Omega}}e^{-i\braket{\kappa \varrho_{-1}, x}}U\, dx\neq0
    \end{align}
    Then, for all $\lambda \in \bbR$ except possibly on a discrete set, we have that, for $\tilde{K}_*\in\{\tilde{K}, \tilde{K}'\}$
    \begin{enumerate}
        \item \label{ConclusionA}There exists an eigenvalue $E_0(\lambda,\tilde{K}_*)$ of multiplicity exactly 2 in $L^2_{\tilde{K}_*}$, with eigenfunctions $\Phi_1(\lambda,x)\in L^2_{\tilde{K}_*,\tau}$, and $\Phi_2(\lambda,x)=\bar{\Phi}_1(\lambda,-x)\in L^2_{\tilde{K}_*,\bar{\tau}} $.
        \item \label{ConclusionB} There is some $\delta_k>0$, and two pairs  $(E_+(\lambda,k),\Phi_+(\lambda,k))$,
        $(E_-(\lambda,k),\Phi_-(\lambda,k))$ - which are Lipchitz continuous in $k$, such that for all $|k-\tilde{K}_*|<\delta$ we have 
        \begin{align*}
            |E_\pm(\lambda,k)-E_0(\lambda,\tilde{K}^*)|^2=|v_d(\lambda)|^2|k-\tilde{K}_*|^2+O(|k-\tilde{K}_*|^3)
        \end{align*}
        So, there is a Dirac cone at $(\tilde{K}_*, E_0(\tilde{K}_*))$.
        \item \label{ConclusionC} The slope of the cone,  $v_d$, is given by 
        \begin{align*}
            v_d(\lambda)=-2i (\Phi_1(\lambda,\cdot),\partial_{x_1} \Phi_2(\lambda,\cdot))
        \end{align*}
    \end{enumerate}
\end{theorem}
It is easy to see that condition (\ref{WFCond}) does not hold in the case of twisted bilayer potentials of the type given in (\ref{AddPotAA}) or (\ref{AddPotAB}): this condition requires that the mode denoted by $\varrho_{-1}$ is not $0$, with respect to the new lattice $\Lambda^\theta$. In other words, we want the Fourier mode corresponding to $k^\theta_{j}$ to be non-zero, for $j\in \{1,2\}$, depending on whether the new periodic lattice is $\Lambda$ or $\Lambda^*$. By duality scaling, one gets that this correspond to $\frac{1}{N}\tilde{k}_j$, where $\tilde{k}\in \{k,v\}$, depending on the underlying lattice. Conversely, $W^\theta_{AA} $ contains twisted potentials (which only twist the Fourier coefficients). The potential's first non-zero mode, in the best-case scenario, corresponds to $k_j$, and thus the lowest frequency $W^\theta_{AA}$ could have is some rotation of $k_j$, and in particular, we have that $\frac{1}{N}\tilde{k}_j$ is not in its support. See the full proof of Proposition \ref{WSupport}.  \par
We note that there are two differences between Theorem \ref{FW}  and our setting: first, we deal with potentials that are periodic with respect to $\Lambda^*$ and not only $\Lambda$, and second, we deal with almost honeycomb potentials, that is not even but invariant under flips (which are given by swapping $(x,\theta)\rightarrow (-x,-\theta)$). These changes will be addressed in the proof of Theorem \ref{Eveneigen}.\par
Thus, we get that we need to extend these results by finding a different condition, and so we prove the following statement
\begin{theorem}\label{DiracPoint}
    Let $\tilde{H}=-\Delta +\lambda U$, for $\lambda\in \bbR$ and $U$ a honeycomb potential, with $\tilde{\Lambda}\in \{\Lambda,\Lambda^*\}$, be such that:
    \begin{align} \label{Zeros}
        \forall \vec{m} \in \calS,\exists \ell \in \bbZ_3, \hat{U}_{\vec{m}-\varrho_\ell}=0
    \end{align}
    Then we may choose $\calS$ such that $\vec{m}-\varrho_{1}\not \in \supp \hat{U}$, with this choice,  if we have 
    \begin{align}\label{ConditionForV}
        \sum_{\vec{m}\in \calS\setminus \{\vec{0}\}}\frac{\hat{U}_{\vec{m}}\hat{U}_{\vec{m}-\varrho_{-1}}}{|\tilde{K}_*|^2-|\tilde{K}_*(\vec{m})|^2}\neq 0 
    \end{align}
    then for all $\lambda \in \bbR$ except possibly on a discrete set, we have that, for $\tilde{K}_*\in\{\tilde{K}, \tilde{K}'\}$ that the conclusions (\ref{ConclusionA}) - (\ref{ConclusionC}) of Theorem \ref{FW} hold.\par
    Furthermore, even if condition (\ref{Zeros}) does not hold, we have that there is some $C>0$ such that 
    \begin{align} \label{asymptotics}
        |v_d(\lambda)|^2\leq C(|\tilde{K}_*|^2+\lambda  \|U\|_\infty +\lambda^2 \|\nabla U\|_{\infty}^2\sum_{\vec{m}\in \calS\setminus \{\vec{0}\}} \frac{1}{|\tilde{K}_*+\tilde{\kappa}\vec{m}|^{4}})+O(\lambda^3\|U\|^3)
    \end{align}
    as $\lambda \rightarrow 0$.
\end{theorem}
The above theorem allows us to conclude our main theorem:
\begin{theorem}\label{DiracPointForTwisted}
    Let $H^\theta=-\Delta +\lambda W^\theta$, for $\lambda\in \bbR$ and twisted bilayer potential with respect to honeycomb potentials $V$ and $Q$, and angle $\theta \in \calC\cap (0,\frac{\pi}{3})$, in either AA or AB stacking. $W^\theta$ is periodic with respect to $\Lambda^\theta$. Let $K_*^\theta \in \bbP^\theta$- one of the points of high symmetry, then if we have
    \begin{align}\label{FWForTBG}
        \hat{W}^\theta_{-\varrho_{-1}}\neq 0
    \end{align}
    or
    \begin{align}\label{ConditionForW}
         \forall \vec{m}\in \calS \exists \ell\in \bbZ_3, \hat{W}^\theta _{\vec{m}-\varrho_\ell}=0 \text{ and }\sum_{\vec{m}\in \calS\setminus \{\vec{0}\}}\frac{\hat{W}^\theta_{\vec{m}}\hat{W}^\theta_{\vec{m}-\varrho_{-1}}}{|K_*^\theta(\vec{m})|^2-|K_*^\theta|^2}\neq 0 
    \end{align}
    Then, for all $\lambda \in \bbR$ except possibly on a discrete set, we have that the conclusions (\ref{ConclusionA}) - (\ref{ConclusionC}) of Theorem \ref{FW} hold.
\end{theorem}
As a result of the proofs above, we get the following result about the vanishing of the Dirac points for small potentials:
\begin{theorem}\label{Vanishing}
    We have for $\theta \in \calC\cap(0,\frac{\pi}{3})$, that for any $\delta>0$,  if $|\lambda|<\frac{\delta}{N^2}$, then there is some  constant $0<C=C(\delta, V,G)$ such that 
    \begin{align*}
        |v_d(\lambda)| \leq \frac{C}{N}+O(N^{-3})
    \end{align*}
\end{theorem}
 Finally, we provide a set of examples for which condition (\ref{ConditionForW}) holds:
\begin{restatable}{proposition}{VConstruction}
\label{VConstruction}
    Define the equivalence relation $\sim_B $ by $\vec{m}\sim_B\vec{n}\iff \exists \ell \in \bbZ_3, B^\ell \vec{m}=\vec{n}$. Then denote  $\tilde{\calS}=\bbZ^2/\sim_B$.\par 
    Let $(a_{\vec{m}})_{\vec{m}\in \tilde{\calS}}$ be exponentially decaying sequence such that $\forall \vec{m}\in \tilde{\calS}, a_{\vec{m}}>0$.
    We define
    \begin{align*}
        V(x)=\pm \sum_{\vec{m}\in \tilde{\calS}}a_{\vec{m}}\sum_{\ell\in \bbZ_3}\cos(\braket{\kappa B^\ell\vec{m},x})
    \end{align*}
    Then $V$ is a honeycomb potential. And if we define the twisted potential as in (\ref{AddPotAA}), with $Q=V$, that is 
    \begin{align*}
        W^\theta=\frac{1}{2}(\calR_{\theta}V+\calR_{-\theta}V)
    \end{align*}
    Then we have that for any  $\theta \in \calC\cap (0,\frac{\pi}{3})$  we have that
    \begin{align*}
         \sum_{\vec{m}\in \calS\setminus \{\vec{0}\}}\frac{\hat{W}^\theta_{\vec{m}}\hat{W}^\theta_{\vec{m}-\varrho_{-1}}}{|K_*^\theta(\vec{m})|^2-|K_*^\theta|^2}\neq 0 ,\text{ and }\forall \vec{m}\in \calS, \exists \ell\in \bbZ_3, \hat{W}^\theta _{\vec{m}-\varrho_\ell}=0 
    \end{align*}
    holds. 
\end{restatable}
\section{Existence of Dirac points }\label{SecondOrderPert}
In this section, we prove Theorem \ref{DiracPoint}- about the existence of Dirac cones. The goal is to get a different set of technical conditions than in \cite{fefferman2012honeycomb}. By the theorem's assumptions, we can choose $\calS$ such that for any $\vec{m}\in \calS$
\begin{align*}
    \hat{U}_{\vec{m}-\varrho_1}=0
\end{align*}
For $\vec{m}=0$, this yields $\hat{U}_{-\varrho_1}=0$, meaning the Fefferman-Weinstein condition (\ref{WFCond}) fails, necessitating a new condition.
Before we prove this new condition we note that Theorem 2.4 in \cite{berkolaiko2018symmetry} uses symmetries to provide an explicit condition for Dirac cones (an eigenvalue of multiplicity two), and so we require the following slight modification:
\begin{theorem}\label{Eveneigen}[A version of Theorem 2.4 in \cite{berkolaiko2018symmetry}]
    Let $\tilde{H}$ be a self-adjoint operator that is periodic with respect to $\Lambda$ or $\Lambda^*$ and invariant under the rotation $R$, and is invariant under $F$ or $F^*$. Let $\tilde{K}_*\in \tilde{\bbP}$ be one of the high symmetry points. Then we have that 
    \begin{align*}
        L^2_{\tilde{K}_*}=L^2_{\tilde{K}_*,1}\oplus L^2_{\tilde{K}_*,\perp}
    \end{align*}
    where the splitting is $\tilde{H}$ invariant. Since $\tilde{H}$ is also invariant under $F$ or $F^*$, we have that all the eigenvalues restricted to $L^2_{\tilde{K}_*,\perp}$ have even multiplicity. If the multiplicity of some eigenvalue $E_0$ is exactly $2$, we have that 
     \begin{align*}
        |E_\pm(\lambda,k)-E_0(\lambda,\tilde{K}_*)|^2=|v_d(\lambda)|^2|k-\tilde{K}_*|^2+O(|k-\tilde{K}_*|^3)
    \end{align*}
    for some $v_d\in \bbC$.
\end{theorem}
\begin{proof}
    There are two differences between this version and, Theorem 2.4 in \cite{berkolaiko2018symmetry}: first, we allow $\Lambda^*$ instead of $\Lambda$, and second, we allow $F^*$ instead of $F$. \par
    The proof of Theorem 2.4 in \cite{berkolaiko2018symmetry} relies only on two steps: First, they show that if $E$ is a double eigenvalue in $L^2_{\tilde{K}_*}$, the conclusion holds (see Lemma 3.1 therein), and the second step shows the splitting and the evenness of the multiplicity (Lemma 4.3). Both lemmas rely only on the symmetries of the Hamiltonian and the restriction to the points of high symmetry subspace (that is, the space $L^2_{\tilde{K}_*}$ is invariant under rotation). So, this theorem can apply to the case where the $U$ is periodic with respect to $\Lambda^*$, with its high symmetry points (irrespective of the choice of base vectors). \par
    To conclude our version of the theorem, we need to address the change between $F$ (which was dealt with in \cite{berkolaiko2018symmetry}) and $F^*$. Replacing $F$ with $F^*$ shifts the symmetry group to act on $(x,\theta)$. Combined with complex conjugation, $F^*$ acts as an involution identical to $\overline{V}$ in the original proof, inducing the same corepresentation structure.
\end{proof}
Now we can prove Theorem \ref{DiracPoint}:
\begin{proof}[Proof of Theorem \ref{DiracPoint}]
    We recall that we consider 
    \begin{align*}
        \tilde{H}=-\Delta +\lambda U
    \end{align*}
    where $U$ is periodic with respect to $\tilde{\Lambda}\in \{\Lambda, \Lambda^*\}$, and $\tilde{\kappa}\in \{\kappa,\nu\}$- the reciprocal lattice matrix. \par 
   
    So, to prove there is a Dirac cone around a point $(\tilde{K}_*, E)$ (or in other words, Theorem \ref{DiracPoint}), we need to show that $E$ has a double eigenvalue and that $v_d\neq 0$.\par
    Following Theorem 2.5 in \cite{berkolaiko2018symmetry} (or Proposition 6.3 in \cite{fefferman2012honeycomb}), at $\lambda=0$ (the free Laplacian), the energy $E=|\tilde{K}_*|^2$ has multiplicity 3. Each space $\{L^2_{\tilde{K}_*,\sigma}\}$ has a simple eigenvalue. By the perturbation theory of simple eigenvalues, each of these eigenvalues extends to an analytic function $E_\sigma(\lambda)$, see \cite{berkolaiko2018symmetry, fefferman2012honeycomb} for more details. Thus it is enough to show that $E_{\tau}=E_{\bar{\tau}}\neq E_{1}$, as functions. By the above, it is sufficient to show that $E_{\tau}(\lambda)\neq E_{1}(\lambda)$ for some $\lambda$ (as the remaining eigenvalues must remain of even multiplicity, and thus have to be of multiplicity $2$). Then, these functions may intersect only on a discrete set. \par
    For this, we consider small $\lambda$ and energies close to $|\tilde{K}_*|^2$, as mentioned above, we have some smooth function $E_\sigma(\lambda)$ such that
    \begin{align*}
        (-\Delta +\lambda U)\Phi_{\sigma} =E_\sigma (\lambda)\Phi_{\sigma}
    \end{align*}
    We recall that for $\lambda=0$, we have that the eigenfunctions in $L^2_{\tilde{K}^*.\sigma}$, for $\sigma \in \{1,\tau,\bar{\tau}\}$ are given by 
    \begin{align*}
        &\psi_0^\sigma=\frac{1}{\sqrt{3}}\sum_{\ell \in \bbZ^3}\sigma^{-\ell }e^{i\braket{\tilde{K}_* +\tilde{\kappa} \rho_\ell, x}}&&\psi_{\vec{m}}^{\sigma}=\frac{1}{\sqrt{3}}\sum_{\ell \in \bbZ^3}\sigma^{-\ell }e^{i\braket{\tilde{K}_*+\tilde{\kappa}(B^\ell\vec{m}+ \rho_\ell), x}}
    \end{align*}
    Using second-order perturbation theory, or the Rayleigh-Schr\"{o}dinger coefficients (see, for example, \cite{RSVol4} -page 7), we get that 
    \begin{align*}
        E_\sigma(\lambda)=E_\sigma(0)+\lambda  E_\sigma^{(1)}+\lambda^2 E_\sigma^{(2)}+O(\lambda^3)
    \end{align*}
    where 
    \begin{align*}
        &E^{(1)}_\sigma= (\psi_0^\sigma,U\psi_0^\sigma)_{L^2_{\tilde{K}_*}}&&E^{(2)}_\sigma=\sum_{\vec{m}\in \calS\setminus \{\vec{0}\}}\frac{|(\psi_{\vec{m}}^\sigma,U\psi_0^\sigma)_{L^2_{\tilde{K}_*}}|^2}{|\tilde{K}_*|^2-|\tilde{K}_*(\vec{m})|^2}
    \end{align*}
    Using the estimate in Theorem 2.1 in \cite{carlsson2024perturbation}, we see that, in fact, we have that 
    \begin{align}\label{UniformExp}
        E_\sigma(\lambda)=E_\sigma(0)+\lambda  E_\sigma^{(1)}+\lambda^2 E_\sigma^{(2)}+O(\lambda^3\|U\|^3)
    \end{align}
    Expanding $U$ into its Fourier series and evaluating the inner product, we use the lattice symmetries $\hat{U}_{B\vec{m}} = \hat{U}_{\vec{m}}$ to obtain:
    \begin{align*}
        (\psi_{\vec{m}}^\sigma,U\psi_0^\sigma)_{L^2_{\tilde{K}_*}} = \frac{1}{3}\sum_{\ell,\ell'\in \bbZ^3}\sigma^{\ell-\ell'}\hat{U}_{\varrho_{\ell'}-B^\ell\vec{m}-\varrho_{\ell}} = \sum_{\ell\in \bbZ^3}\sigma^{-\ell}\hat{U}_{\vec{m}-\varrho_{\ell}} = \hat{U}_{\vec{m}}+\sigma \hat{U}_{\vec{m}-\varrho_{-1}}
    \end{align*}
    where we used that $\hat{U}_{B\vec{m}}=\hat{U}_{\vec{m}}$ for all $\vec{m}\in \bbZ^2$, and we recall that we chose that $\calS$ in such a way that $\hat{U}_{\vec{m}-\varrho_1}=0$, for all $\vec{m}\in \calS$.\par
    In particular, we get that 
    \begin{align*}
        E^{(1)}_\sigma=(\psi_0^\sigma ,U\psi_0^\sigma)_{L^2_{\tilde{K}_*}}& =\hat{U}_{0}+\sigma \hat{U}_{0-\varrho_{-1}}
    \end{align*}
   Note that
    \begin{align*}
        &B^{-1}\varrho_1=-\varrho_{-1} \implies \hat{U}_{-\varrho_{-1}}=\hat{U}_{-\varrho_1}=0
    \end{align*}
    since $\hat{U}_{B\vec{m}}=\hat{U}_{-\vec{m}}$ for all $\vec{m}\in \bbZ^2$. So we got that 
    \begin{align*}
        E^{(1)}_\sigma=\hat{U}_{\vec{0}}
    \end{align*}
    So $E^{(1)}_\sigma$ is independent of $\sigma$ - so we see that the eigenvalues do not separate in the first order (as expected from this argument in \cite{berkolaiko2018symmetry} or \cite{fefferman2012honeycomb}).\par
    We compute the next order, starting by noting that
    \begin{align*}
        |(\psi_{\vec{m}}^\sigma,U\psi_0^\sigma) _{\tilde{\Omega}}|^2=|\hat{U}_{\vec{m}}|^2+|\hat{U}_{\vec{m}-\varrho_{-1}}|^2+(\sigma+\sigma^{-1})\hat{U}_{\vec{m}}\hat{U}_{\vec{m}-\varrho_{-1}}
    \end{align*}
    So we have that 
    \begin{align*}
        E^{(2)}_\sigma=\sum_{\vec{m}\in \calS\setminus \{\vec{0}\}}\frac{|\hat{U}_{\vec{m}}|^2+|\hat{U}_{\vec{m}-\varrho_{-1}}|^2+(\sigma+\sigma^{-1})\hat{U}_{\vec{m}}\hat{U}_{\vec{m}-\varrho_{-1}}}{|\tilde{K}_*|^2-|\tilde{K}_*(\vec{m})|^2}
    \end{align*}
    Now by assumption 
    \begin{align*}
        \sum_{\vec{m}\in \calS\setminus \{\vec{0}\}}\frac{\hat{U}_{\vec{m}}\hat{U}_{\vec{m}-\varrho_{-1}}}{|\tilde{K}_*|^2-|\tilde{K}_*(\vec{m})|^2}\neq 0
    \end{align*}
    Thus,
    \begin{align*}
        E_\tau(\lambda) \neq E_1(\lambda)
    \end{align*}
    as they differ in the second-order term. By Theorem \ref{Eveneigen}
     we conclude that the multiplicity is even in $L^2_{\tilde{K}_*,\perp}$, and so we can conclude that 
    \begin{align*}
        E_{\bar{\tau}}(\lambda) =E_\tau(\lambda) \neq E_1(\lambda)
    \end{align*}
    as needed.\par
    Finally, we need to show that $v_d$ is non-zero except possibly on a discrete set. Since $v_d(\lambda)$ is analytic, it suffices to show $v_d(0) \neq 0$. Because the unperturbed operator ($\lambda=0$) is the free Laplacian, a standard computation yields, with our conventions (similar to the one done in \cite[Proposition 6.3]{fefferman2012honeycomb} or in  \cite[Theorem 2.5]{berkolaiko2018symmetry}):
    \begin{align*}
        v_d(0) =-2i(\psi_0^\tau,\partial_{x_1}\psi_0^{\bar{\tau}})_{L^2_{\tilde{K}_*}}= \begin{cases}
            -4\pi ,&\tilde{\kappa}=\kappa\\
            -\sqrt{3}i,&\tilde{\kappa}=\nu
        \end{cases}\neq 0
    \end{align*}
    as needed. \par
    For the last part of the theorem, we recall that 
    \begin{align*}
        v_d(\lambda) =-2i(\Phi_1,\partial_{x_1}\Phi_2)_{L^2_{\tilde{K}_*}}
    \end{align*}
    where $\Phi_j$ are the eigenfunctions which have
    \begin{align*}
        (-\Delta+U)\Phi_j=E_\sigma(\lambda)\Phi_j
    \end{align*}
    for $j\in \{1,2\}$, where $\sigma\in \{\tau,\bar{\tau}\}$. So we can write
    \begin{align*}
        |v_d(\lambda)|^2=4|(\Phi_1,\partial_x \Phi_2)_{L^2_{\tilde{K}_*}}|^2\leq 4\|\Phi_1\|^2\|\partial_x\Phi_2\|^2
    \end{align*}
   With the normalization of the eigenfunctions ($\|\Phi_1\|=1=\|\Phi_2\|$), we get
    \begin{align*}
        |v_d(\lambda)|^2\leq 4\|\partial_x\Phi_2\|^2\leq 4\|\nabla\Phi_2\|^2
    \end{align*}
    Recalling $E_\tau=E_{\bar{\tau}}=E$,  we write
    \begin{align*}
        \|\nabla\Phi_2\|^2 = (\Phi_2,( E-\lambda U)\Phi_2)_{L^2_{\tilde{K}_*}} \leq E + |\lambda| \|U\|_\infty
    \end{align*}
    Using (\ref{UniformExp}) for $E=E_\tau(\lambda)$ and applying the bound $|\hat{U}_{\vec{m}}| \leq \frac{\|\nabla U\|_\infty}{ |\tilde{K}_*+\kappa \vec{m}|}$ (as $U$ is smooth) directly to the second-order term yields:
    \begin{align*}
        |v_d(\lambda)|^2 &\leq 4(|E_\tau(0)|+|\lambda|  |E_\tau^{(1)}|+\lambda^2 |E_\tau^{(2)}|+O(\lambda^3\|U\|^3)+2|\lambda| \|U\|_\infty)\\
        &\leq 4\left(|\tilde{K}_*|^2 + 3|\lambda|\|U\|_\infty + 4\lambda^2 \|\nabla U\|_{\infty}^2\sum_{\vec{m}\in \calS\setminus \{\vec{0}\}} \frac{1}{|\tilde{K}_*+\tilde{\kappa}\vec{m}|^{4}}\right) + O(\lambda^3\|U\|^3)
    \end{align*}
    which yields to the claimed bound.\par
  In the more general case, where we do not assume that $\hat{U}_{\vec{m}-\varrho_1}=0$ for all $\vec{m}\in \calS$, we still have that
    \begin{align*}
        &|v_d(\lambda)|^2\leq 4(|E_\tau(0)|+|\lambda  ||E_\tau^{(1)}|+\lambda^2 |E_\tau^{(2)}|+O(\lambda^3\|U\|^3)+2|\lambda |\|U\|_\infty)
    \end{align*}
    In this case, however, we have that
    \begin{align*}
        |E^{(1)}_\tau|=|(\psi_0^\tau ,U\psi_0^\tau)_{L^2_{\tilde{K}_*}}|& =|\hat{U}_{0}+(\tau+\bar{\tau}) \hat{U}_{0-\varrho_{-1}}|\leq 3\|U\|_\infty  
    \end{align*}
    And
    \begin{align*}
        &|E^{(2)}_\tau|\leq \sum_{\vec{m}\in \calS\setminus \{\vec{0}\}}\frac{|\sum_{\ell\in \bbZ^3}\tau^{-\ell} \hat{U}_{\vec{m}-\varrho_\ell}|^2}{||\tilde{K}_*|^2-|\tilde{K}_*(\vec{m})|^2|}\leq 9\|\nabla U\|^2\sum_{\vec{m}\in \calS\setminus \{\vec{0}\}}\frac{1}{|\tilde{K}_*(\vec{m})|^4}
    \end{align*}
    So we get that there is some constant $C>0$ such that
    \begin{align*}
        |v_d(\lambda)|^2\leq& C(|\tilde{K}_*|^2+|\lambda|  \|U\|_\infty +\lambda^2 \|\nabla U\|_{\infty}^2\sum_{\vec{m}\in \calS\setminus \{\vec{0}\}} \frac{1}{|\tilde{K}_*+\tilde{\kappa}\vec{m}|^{4}})+O(\lambda^3\|U\|^3)
    \end{align*} 
    as needed.
\end{proof}
\begin{remark}\label{ConditionRemark}
    We note that we could have that $\hat{U}_{\vec{m}}\hat{U}_{\vec{m}-\varrho_{-1}}=0$ for any $\vec{m}\in \supp \hat{U}^\theta$. As the proof above shows, one must go to higher-order terms in the perturbation series to get a sufficient non-degeneracy condition for such cases. We will not develop the other terms in this work. Such consideration might also affect the asymptotic results for $v_d(\lambda)$.
\end{remark}
\section{Twisted bilayer potential}\label{TBP}
This section proves two of our main results. First, we prove Theorem \ref{NewLattice}, which describes the commensurate angles. We then turn our attention to the representative example of (\ref{AddPotAA}) or (\ref{AddPotAB}), and prove Lemma \ref{WDescription} describing the Fourier support of $W^\theta$, which, together with Theorem \ref{DiracPoint}, allows us to prove Theorem \ref{DiracPointForTwisted}, about the existence of Dirac cones for twisted potentials, in full. 
\subsection{Proof of Theorem \ref{NewLattice}}
We start by providing a full description of the commensurate angles. We mention that a different approach to finding the new lattice vectors can be found in  \cite{rodriguez2015coincidence} using Clifford algebras. However, their representation is highly fragmented. They derive a different basis for each arithmetic case (depending on parity and whether $3 \mid a$), and the resulting spanning vectors differ from the representation we provide here. Another algebraic approach to the problem is given in \cite{baake1997solution, baake2006multiple}, where they describe the possible coincidence lattices in full but do not give a full description of the primitive vectors. So, we provide complete proof that the new lattice is periodic with respect to a scaled version of the honeycomb lattice. \par
First, we show that we can reduce our problem to the range $\theta\in [0,\frac{\pi}{3})$:
\begin{proposition}\label{symmetry}
    For any $\theta \in [0,2\pi)$, there exists some $\tilde{\theta}\in [0,\frac{\pi}{3}]$ such that $H^{\theta}=H^{\tilde{\theta}}$.   
\end{proposition}
\begin{proof}
    Let $\theta \in [0,2\pi)$.  We start by noting that 
    \begin{align*}
        V(R_{-\theta-\frac{\pi}{3}}x)=V(R_{-\frac{\pi}{3}}R_{-\theta}x)=V(RR_{\pi}^{-1}R_{-\theta}x)=V(-R_{-\theta}x)=V(R_{-\theta}x)
    \end{align*}
    Since $R_{\pi}=-\id$, and $R=R_{\frac{2\pi}{3}}$, the same also holds for $Q$.\par 
    Similarly, we can get $V(R_{\theta+\frac{\pi}{3}}x)=V(R_{\theta}x)$, so we get that 
    \begin{align*}
        &W^{\theta+\frac{\pi}{3}}_{AA}(x)=G(\calR_{\theta+\frac{\pi}{3}}V,\calR_{-\theta-\frac{\pi}{3}}Q)=G(\calR_{\theta}V,\calR_{-\theta}Q)=W^{\theta}_{AA}(x)
    \end{align*}
    and similarly for $W^{\theta+\frac{\pi}{3}}_{AB}$.\par
    As the potentials are the same. So we conclude it is enough to take $\theta\in [0,\frac{\pi}{3})$. 
\end{proof}
Now, we give a better description of the commensurate lattice $\Lambda^\theta=R_{\theta}\Lambda \cap R_{-\theta}\Lambda$
\begin{lemma}\label{DiscrptionofLambda}
    Let $\theta \in\calC\cap (0,\frac{\pi}{3})$, and $\Lambda^\theta=R_{\theta}\Lambda\cap R_{-\theta}\Lambda$. Then we have the following 
    \begin{enumerate}
        \item \label{Rationality} First $\tan(\theta)=\frac{\sqrt{3}b}{a}$, for $0<b<a$, and $a$ and $b$ are co-primes.
        \item\label{LatticeDesc} Denoting  
        \begin{align*}
            &\alpha =\begin{cases}
                8\pi, & 3\mid a \text{ and } 2\nmid ab\\
                2, & 3\nmid a \text{ and } 2\nmid ab\\
                4\pi , & 3\mid a \text{ and } 2\mid ab\\
               1, & 3\nmid a \text{ and } 2\mid ab
            \end{cases},&N=\frac{1}{\alpha} \sqrt{a^2+3b^2}
        \end{align*}
        Then we have that
        \begin{align*}
            \Lambda^\theta=N\begin{cases}
                \Lambda, & 3\nmid a\\
                \Lambda^* & 3\mid a
            \end{cases}
        \end{align*}
        \item And we have that $R_{\theta}=\frac{1}{\alpha N}\begin{pmatrix}
                a&-\sqrt{3}b\\
                \sqrt{3}b &a
            \end{pmatrix}$.
    \end{enumerate}
\end{lemma}
Before we prove this claim, we need the following identity:
\begin{proposition}\label{EqualitofLattices}
    We have that 
    \begin{align*}
        \bigcup_{r\in \{0,\pm1\}}(4\pi)(\bbZ^2+\frac{ r}{3}\begin{pmatrix}1\\1\end{pmatrix})=\nu^{-1}\kappa \bbZ^2
    \end{align*}
\end{proposition}
\begin{proof}
    We recall that we have 
    \begin{align*}
        &\frac{1}{4\pi}\kappa =\begin{pmatrix}\frac{1}{2\sqrt{3}} &\frac{1}{2\sqrt{3}}\\\frac{1}{2} &-\frac{1}{2}\end{pmatrix}  &\nu^{-1}=\begin{pmatrix}\frac{1}{\sqrt{3}} &1\\\frac{1}{\sqrt{3}} &-1\end{pmatrix}
    \end{align*}
    One can compute and get $4\pi(\nu^{-1}\kappa )^{-1}=\begin{pmatrix}2&1\\1&2\end{pmatrix}$. So, we need to show that 
    \begin{align*}
        \bbZ^2=\begin{pmatrix}2&1\\1&2\end{pmatrix}\bigcup_{r\in \{0,\pm1\}}(\bbZ^2+\frac{r}{3}\begin{pmatrix}1\\1\end{pmatrix})=\bigcup_{r\in \{0,\pm1\}}(\begin{pmatrix}2&1\\1&2\end{pmatrix}\bbZ^2+r\begin{pmatrix}1\\1\end{pmatrix})
    \end{align*}
    First, that the LHS is a subset of $\bbZ^2$ is immediate. Now, we show the other inclusion: Let $\begin{pmatrix}m\\n\end{pmatrix}\in \bbZ^2$, then we take $r\in \bbZ_3$ such that $r\equiv 2m-n\mod 3$. Define
    \begin{align*}
        &\tilde{m}=\frac{2m-n-r}{3}, &\tilde{n}=n-m+\tilde{m}=\frac{2n-m-r}{3}
    \end{align*}
    We note that $\tilde{m},\tilde{n}\in \bbZ^2$, and so we have that $2\tilde{m}+\tilde{n}+r=m,2\tilde{n}+\tilde{m}+r=n$, so we can write 
    \begin{align*}
        \begin{pmatrix}m\\n\end{pmatrix} =\begin{pmatrix}2&1\\1&2\end{pmatrix} \begin{pmatrix}\tilde{m}\\\tilde{n}\end{pmatrix}+r\begin{pmatrix}1\\1\end{pmatrix}
    \end{align*}
    which gives us the reverse inclusion and allows us to conclude 
    \begin{align*}
        \bigcup_{r\in \{0,\pm1\}}(\begin{pmatrix}2&1\\1&2\end{pmatrix}\bbZ^2+r\begin{pmatrix}1\\1\end{pmatrix})= \bbZ^2
    \end{align*}
    which concludes the proof of the proposition.
\end{proof}
Now we can prove Lemma \ref{DiscrptionofLambda} describing the new lattice generated by a commensurate angle:
\begin{proof}[Proof of Lemma \ref{DiscrptionofLambda}]
    We recall that 
    \begin{align*}
        R_{\theta}+R_{-\theta}=2\cos(\theta)\Id
    \end{align*}
    So if $x\in R_{\theta}\Lambda\cap R_{-\theta}\Lambda$ we have some $u,v \in \Lambda$ such that, since $0<\theta<\frac{\pi}{3}$, $\cos(\theta) \neq 0$
    \begin{align*}
        &x= R_{\theta}u=R_{-\theta}v=2\cos(\theta)v-R_{\theta}v\implies v=\frac{1}{2\cos(\theta)}R_{\theta}(u+v)\\
        &x=\frac{1}{2\cos(\theta)}(u+v)\in \frac{1}{2\cos(\theta)}\Lambda 
    \end{align*}
    In particular, we get that 
    \begin{align*}
       R_{\theta}\Lambda \cap R_{-\theta}\Lambda  \subset\frac{1}{2\cos(\theta)} \Lambda\cap  R_{\theta}\Lambda 
    \end{align*}
    Denoting $A=\nu^{-1}R_{\theta}\nu $ we get that
    \begin{align*}
        (R_{\theta}\Lambda \cap R_{-\theta}\Lambda  )\subset \nu A(\frac{1}{2\cos(\theta)}A^{-1}\bbZ^2\cap \bbZ^2) 
    \end{align*}
    So, we may compute 
    \begin{align*}
        A&=\begin{pmatrix}
            \cos(\theta) +\frac{\sin(\theta)}{\sqrt{3}}&\frac{2}{\sqrt{3}}\sin(\theta)\\-\frac{2}{\sqrt{3}}\sin(\theta)&\cos(\theta) -\frac{\sin(\theta)}{\sqrt{3}}
        \end{pmatrix}\\
         \frac{1}{\cos(\theta)}A^{-1}&=\Id - \frac{\tan(\theta)}{\sqrt{3}}\begin{pmatrix} 1&2\\-2&-1\end{pmatrix}=\Id- \frac{\tan(\theta)}{\sqrt{3}}\calI
    \end{align*}
    Thus, we get that 
    \begin{align*}
        R_{\theta}\Lambda \cap R_{-\theta}\Lambda  \subset \nu A((\frac{1}{2}\Id- \frac{\tan(\theta)}{2\sqrt{3}}\calI)\bbZ^2\cap \bbZ^2) 
    \end{align*}
    We note that since $R_{\theta}\Lambda \cap R_{-\theta}\Lambda \neq \{0\}$, then we have some $(\Id+ \frac{\tan(\theta)}{2\sqrt{3}}\calI)\bbZ^2\cap \bbZ^2\neq \{0\}$. Thus, we have some $u,v\in \bbZ^2$ such that 
    \begin{align*}
        &\bbZ^2\ni v=\frac{1}{2}u- \frac{\tan(\theta)}{2\sqrt{3}}\calI u \implies \frac{\tan(\theta)}{\sqrt{3}}\calI u\in \bbZ^2 
    \end{align*}
    So, we conclude that $\bbZ^2\cap \frac{\tan(\theta)}{\sqrt{3}}\calI \bbZ^2\neq \{0\}$, thus we conclude that, in particular, $\frac{\tan(\theta)}{\sqrt{3}} \in \bbQ$.
    So we write that  $\frac{\tan(\theta)}{\sqrt{3}}=\frac{b}{a}$, where $a,b\in \bbZ$ are co-prime. Since $0<\theta<\frac{\pi}{3}$, we get that 
    \begin{align*}
        0<\frac{b}{a}=\frac{\tan(\theta)}{\sqrt{3}}<1
    \end{align*}
    and we can choose $a,b>0$, and $b<a$, thus proving part \ref{Rationality} of the Lemma.\par
    In particular we got that if $v\in (\Id- \frac{b}{2a}\calI)\bbZ^2\cap \bbZ^2$ we have some $u\in \bbZ^2$ such that 
    \begin{align*}
        &v=(\frac{1}{2}\Id- \frac{b}{2a}\calI) u\implies 2av=au-b\calI u
    \end{align*}
    From the equality above, since $\calI\equiv \Id \mod 2$ and $ \gcd(a,b)=1$, we can get the following congruences: 
    \begin{align*}
        &0\equiv (a+b)u \mod 2, &&0\equiv \calI u \mod a
    \end{align*}
    Solving this linear congruence system over $\mathbb{Z}^2$ requires carefully tracking the parity and divisibility of $a$ and $b$. So we delegate the solution to Appendix \ref{ExpLatCom}, the result of which is that:
    \begin{align*}
        &\exists r\in \bbZ_3, p \in \bbZ^2,u = 2^{1-\epsilon} a \left(p+\frac{\rho r}{3^{\rho}}\begin{pmatrix}1\\1\end{pmatrix}\right)\implies   v= \frac{1}{2\cos(\theta)}A^{-1}2^{1-\epsilon}  a (p+\frac{\rho r}{3^{\rho}}\begin{pmatrix}1\\1\end{pmatrix})
    \end{align*}
    for $\epsilon=\begin{cases} 1,&2\nmid ab\\ 0,& 2\mid ab \end{cases} $ and $\rho=\begin{cases}1,& 3\mid a\\0,& 3\nmid a\end{cases} $.
    Recall that $(R_{\theta}\Lambda \cap R_{-\theta}\Lambda  )\subset \nu A(\frac{1}{2\cos(\theta)}A^{-1}\bbZ^2\cap \bbZ^2)$. Thus, for any $\textbf{a}\in (R_{\theta}\Lambda \cap R_{-\theta}\Lambda)$, there exist some $r\in \bbZ_3$ and $u \in \bbZ^2$ such that
    \begin{align*}
        \textbf{a}=\nu A\frac{1}{2\cos(\theta)}A^{-1}2^{1-\epsilon} a (u+\frac{\rho r}{3^{\rho}}\begin{pmatrix}1\\1\end{pmatrix})=\nu \frac{1}{2^\epsilon\cos(\theta)}a (u+\frac{\rho r}{3^{\rho}}\begin{pmatrix}1\\1\end{pmatrix})
    \end{align*}
    We note that since $0<\theta <\frac{\pi}{3}\implies \cos(\theta)>0$, we can write
    \begin{align*}
        \cos(\theta)=\frac{1}{\sqrt{1+\tan^2(\theta)}}=\frac{a}{\sqrt{a^2+3b^2}}
    \end{align*}
    Denote $N=\sqrt{a^2+3b^2} 2^{-\epsilon}(4\pi)^{-\rho}$, we get that 
    \begin{align*}
         A\bbZ^2 \cap A^{-1}\bbZ^2   \subset  \bigcup_{r\in \bbZ_3} (4\pi)^{\rho} N (\bbZ^2+\frac{\rho r}{3^{\rho}}\begin{pmatrix}1\\1\end{pmatrix})
    \end{align*}
    We show the opposite containment: Let $p\in \bbZ^2,r \in \bbZ_3$, and let
    \begin{align*}
        v=N(4\pi)^{\rho} (p+\frac{\rho r}{3^{\rho}}\begin{pmatrix}1\\1\end{pmatrix})
    \end{align*}
    We note that we have that 
    \begin{align*}
        A=\cos(\theta)(\Id+\frac{\tan(\theta)}{\sqrt{3}} \calI)=\frac{1}{N2^{\epsilon}(4\pi)^{\rho}}\begin{pmatrix}a+b&2b\\-2b&a-b\end{pmatrix}
    \end{align*}
    So we have that 
    \begin{align*}
        Av&=\frac{1}{2^{\epsilon}}\begin{pmatrix}a+b&2b\\-2b&a-b\end{pmatrix} (p+\frac{\rho r}{3^{\rho}}\begin{pmatrix}1\\1\end{pmatrix})\\
        &=\begin{pmatrix}
            2^{-\epsilon}(a+b) &2^{1-\epsilon} b\\-2^{1-\epsilon} b&2^{-\epsilon}(a-b) 
        \end{pmatrix}p+\rho r\begin{pmatrix}
            2^{-\epsilon}3^{-\rho }(a+3b)\\2^{-\epsilon}3^{-\rho }(a-3b) 
        \end{pmatrix}
    \end{align*}
    Noting that $2^{-\epsilon}(a\pm b), 2^{-\epsilon}3^{-\rho }(a\pm 3b), \text{ and } 2^{1-\epsilon}$ are all integers, we conclude that $Av\in \bbZ^2$. A similar computation (up to substituting $b\mapsto -b$) implies that $A^{-1}v\in \bbZ^2$, which gives the opposite containment.\par
    Thus, we may conclude 
    \begin{align*}
        A^{-1}\bbZ^2\cap A\bbZ^2=(4\pi)^{\rho} N\bigcup_{r\in \bbZ_3} (\bbZ^2+\frac{\rho r}{3^\rho}\begin{pmatrix}1\\1\end{pmatrix})
    \end{align*}
    Using the identity in Proposition \ref{EqualitofLattices}, we can conclude that 
    \begin{align}\label{Aequa}
         A^{-1}\bbZ^2\cap A\bbZ^2=N\begin{cases}
             \bbZ^2, &\rho=0\\
             \nu^{-1}\kappa \bbZ^2,&\rho=1
         \end{cases} \implies R_{\theta}\Lambda \cap R_{-\theta}\Lambda=N\begin{cases}
           \Lambda , & 3\nmid a\\
           \Lambda^* ,& 3\mid a
       \end{cases}
    \end{align}
    We have shown part \ref{LatticeDesc} of the Lemma, for $\alpha= 2^{\epsilon}(4\pi )^{\rho}$. \par
    Finally, we note that
    \begin{align*}
        R_{\theta}&=\begin{pmatrix}
            \frac{\sqrt{3}}{2}&\frac{\sqrt{3}}{2}\\
            \frac{1}{2}& -\frac{1}{2}
        \end{pmatrix} \frac{1}{\sqrt{a^2+3b^2}}\begin{pmatrix}
            a+b&2b\\-2b&a-b
        \end{pmatrix}\begin{pmatrix}
            \frac{1}{\sqrt{3}} &1 \\ \frac{1}{\sqrt{3}} &-1
        \end{pmatrix}=\frac{1}{2^{\epsilon}(4\pi)^{\rho}N }\begin{pmatrix}
            a&-\sqrt{3}b\\
            \sqrt{3}b &a
        \end{pmatrix}
    \end{align*}
    as claimed, concluding the proof of the lemma. 
\end{proof}
This description allows us to conclude that there are no rational rotations in $\calC\cap (0,\frac{\pi}{3})$, other than $\frac{\pi}{6}$:
\begin{corollary}\label{rationality}
    Let $\theta\in (0,\frac{\pi}{3})\setminus \{\frac{\pi}{6}\}$ such that $\frac{\tan(\theta)}{\sqrt{3}}\in \bbQ$ Then $\theta \not \in \pi \bbQ$. 
\end{corollary}
\begin{proof}
    Since we have that $\frac{\tan(\theta)}{\sqrt{3}}\in \bbQ\implies \tan^2(\theta)\in \bbQ$ by the generalization of Niven's Theorem found in \cite{nunn2021proof}, we have that $\theta\in \bbQ \pi $ only if $\theta$ is a integer multiple of $\frac{\pi}{4},\frac{\pi}{6}$, which is not in the domain above. 
\end{proof}
\begin{remark}
    We note that for $\theta=\frac{\pi}{6}$, for AA stacking,  we have that 
    \begin{align*}
        W_{AA}^{\frac{\pi}{6}}(x)&=G( \calR_{\frac{\pi}{6}}V,\calR_{-\frac{\pi}{6}}Q)=G( \calR_{\frac{\pi}{6}}V,\calR_{\frac{\pi}{6}}Q)=\calR_{\frac{\pi}{6}}G(V,Q)= \calR_{\frac{\pi}{6}}W_{AA}^0
    \end{align*}
    So we have that $H_{AA}^{\frac{\pi}{6}}$ is uniterily equivalent (by rotation by $\frac{\pi}{6}$) to $H^0_{AA}$. Moreover, we have that $R_{\frac{\pi}{6}}\Lambda=\frac{\sqrt{3}}{4\pi}\Lambda^* $, as predicted by Lemma \ref{DiscrptionofLambda}.\par
    For AB stacking, we have that
    \begin{align*}
        W_{AB}^{\frac{\pi}{6}}(x)&=G^*( \calR_{\frac{\pi}{6}}\frac{1}{3}\sum_{j=-1}^1\calT_{-\frac{1}{2}R^jR_{-\frac{\pi}{6}}K_0}V,\calR_{\frac{\pi}{6}}\frac{1}{3}\sum_{j=-1}^1\calT_{\frac{1}{2}R^jR_{-\frac{\pi}{6}}K_0}V)
    \end{align*}
     And we note that 
    \begin{align*}
        R_{-\frac{\pi}{6}}K_0&=\begin{pmatrix}
            \frac{1}{2}\\ \frac{1}{2\sqrt{3}}
        \end{pmatrix}=\frac{\sqrt{3}}{4\pi} (\frac{1}{3}(k_2-k_1)+k_1)\in \frac{\sqrt{3}}{4\pi} (\Lambda^*+K_0^*)
    \end{align*}
    Thus, we get that a twist of $\frac{\pi}{6}$ of AB stacking results in AB stacking of the dual lattice, up to a scaling factor of $\frac{\sqrt{3}}{4\pi} $, as predicted by Lemma \ref{DiscrptionofLambda}.
\end{remark}
\subsection{Existence of Dirac points for additive twisted bilayer potentials}
In the following section,  we consider specifically 
\begin{align*}
    &W^\theta_{0,AA}=\frac{1}{2}(\calR_{\theta}V+ \calR_{-\theta}Q)\\
    &W^\theta_{0,AB}=\frac{1}{2}(\frac{1}{3}\sum_{j=-1}^1\calT_{-\frac{1}{2}R^jK_0}\calR_{\theta}V+ \frac{1}{3}\sum_{j=-1}^1\calT_{\frac{1}{2}R^jK_0}\calR_{-\theta}V)
\end{align*}
For this potential, we establish some results relating to the support of $(\hat{W}_{0, AA}^\theta)_{\vec{m}},(\hat{W}_{0, AB}^\theta)_{\vec{m}}$. This section consider $W^{\theta}_{0,AA}$ a twisted bilayer potential for $\theta \in \calC\cap (0,\frac{\pi}{3})$, and denote 
\begin{align*}
    &\calA_1=( N\kappa^\theta)^{-1} R_{\theta}\kappa, &\calA_{-1}=( N\kappa^\theta)^{-1} R_{-\theta}\kappa
\end{align*}
where we recall that $N\kappa^\theta \in \{\kappa,\nu\}$.  \par 
We start by computing $\calA_1$ explicitly, replacing $b\mapsto -b$ leads to $\calA_{-1}$. For that, we first note that, for $\alpha $ as in Proposition \ref{DiscrptionofLambda} 
\begin{align*}
    &R_{\theta}\kappa=\frac{4\pi }{\sqrt{3}\alpha N}\begin{pmatrix}\frac{a-3b}{2} & \frac{a+3b}{2} \\\frac{\sqrt{3}(a+b)}{2}& \frac{\sqrt{3}(b-a)}{2}\end{pmatrix}
\end{align*}
So, we compute if $N\kappa^\theta=\kappa$
\begin{align*}
   &(N\kappa^\theta)^{-1} R_{\theta}\kappa=\frac{1}{\alpha N}\begin{pmatrix}
           a-b & 2b\\ -2b &a+b
       \end{pmatrix}=(A^{-1})^T
\end{align*}
then we have that $\det \calA =1$.\par
And, if $N\kappa^\theta=\nu$ 
\begin{align*}
   &(N\kappa^\theta)^{-1} R_{\theta}\kappa=\frac{1}{N 3 \cdot 2^{\epsilon}}\begin{pmatrix}2a & -a+3b\\ -a-3b &2a\end{pmatrix}
\end{align*}
We note that the last expression is, up to a factor of $N$, an integer matrix, since $3\mid a$. And we have that $\det \calA=(\frac{4\pi}{\sqrt{3}})^2$. \par 
The above notation allows us to provide more details on the Fourier support of $W^\theta_{0, AA}, W^\theta_{0, AB}$.
\begin{proposition}\label{WSupport}
    We have that for $W^\theta_{0,AA},W^\theta_{0,AB}$ as above, for $\theta\in \calC\cap (0,\frac{\pi}{3})$, and $x\in \{AA,AB\}$
    \begin{align*}
        &\supp \hat{W}_{0,x}^\theta=\{\vec{m}\mid (\hat{W}_{0,x})^\theta_{\vec{m}}\neq 0\}\subset N(\calA_1\bbZ^2\cup \calA_{-1}\bbZ^2)
    \end{align*}
\end{proposition}
\begin{proof}
    We start by focusing on the AA stacking, and we note that we know that 
    \begin{align*}
        &V(x)=\sum_{\vec{p}\in \bbZ^2}\hat{V}_{\vec{p}}e^{i\braket{\kappa \vec{p},x}},&Q(x)=\sum_{\vec{p}\in \bbZ^2}\hat{Q}_{\vec{p}}e^{i\braket{\kappa \vec{p},x}}
    \end{align*}
    So we have that
    \begin{align*}
        W_{0,AA}^\theta(x)&=\frac{1}{2}(\sum_{p\in \bbZ^2}\hat{V}_{\vec{p}}e^{i\braket{R_{\theta}\kappa\vec{p},x}}+\sum_{p\in \bbZ^2}\hat{Q}_{\vec{p}}e^{i\braket{R_{-\theta}\kappa\vec{p},x}})
    \end{align*}
    We note that $R_{\pm\theta}\kappa\vec{p}=N\kappa^\theta \calA_{\pm1}\vec{p}$, so we write
    \begin{align*}
        &W_{0,AA}^\theta(x)=\frac{1}{2}(\sum_{\vec{p}\in \bbZ^2}\hat{V}_{\vec{p}}e^{i\braket{\kappa^\theta N\calA_1\vec{p},x}}+\sum_{\vec{p}\in \bbZ^2}\hat{Q}_{\vec{p}}e^{i\braket{\kappa^\theta N\calA_{-1}\vec{p},x}})
    \end{align*}
    So we got that 
    \begin{align}\label{FourierId}
        W_{0,AA}^\theta(x)&=\frac{1}{2}(\sum_{\vec{q}\in N \calA_1 \bbZ^2}\hat{V}_{\frac{1}{N}\calA_1^{-1}\vec{q}}e^{i\braket{\kappa^\theta\vec{q},x}}+\sum_{\vec{q}\in N\calA_{-1} \bbZ^2}\hat{Q}_{\frac{1}{N}\calA_{-1}^{-1}\vec{q}}e^{i\braket{\kappa^\theta\vec{q},x}})
    \end{align}
    On the other hand, we have that, as a function periodic with respect to $\Lambda^\theta$: 
    \begin{align*}
        W_{0,AA}^{\theta}(x)=\sum_{\vec{m}\in \bbZ^2}(\hat{W}_{0,AA}^\theta)_{\vec{m}}e^{i\braket{\kappa^\theta \vec{m},x}}
    \end{align*}
    So, we may conclude  
    \begin{align*}
        \supp \hat{W}_{0,AA}^\theta\subset N( \calA_1\bbZ^2\cup \calA_{-1}\bbZ^2)
    \end{align*}
    as claimed. \par
    The argument for AB stacking is identical since a shift in real space corresponds to multiplying $V_{\vec{p}}$ by the phase, e.g., $e^{i\braket{\kappa \vec{p}, K_0}}$- which does not change the Fourier support. 
\end{proof}
Now we show that $\varrho_{-1}$ can be decomposed into the two lattices:
\begin{corollary}\label{DescribeRho}
    We have that $\varrho_{-1}\in N\calA_1\bbZ^2+N\calA_{-1}\bbZ^2$.
\end{corollary}
\begin{proof}
    We start in the case where $N\kappa^\theta=\kappa$. We need to show that 
    \begin{align*}
        \begin{pmatrix}-1\\0\end{pmatrix}\in N\calA_1\bbZ^2+N\calA_{-1}\bbZ^2
    \end{align*}
    In this case, we have that $3\nmid a$. Then, we note that we have that $(a,b)$ and $(a,3)$ are both pairs of co-prime numbers. So we have some numbers $\tilde{p},\tilde{q},m,n\in \bbZ$ such that
    \begin{align}
        &a\tilde{p}+b\tilde{q}=1\label{abBez}\\
        &3m+an=1\label{a3Bez}
    \end{align}
    by Bézout's identity theorem. Denote $q=\tilde{q}+a(\tilde{p}+\tilde{q}),p=\tilde{p}-b(\tilde{p}+\tilde{q}) $, we note that then we have that 
    \begin{align}
        ap+bq=1\label{abBezNotilde}
    \end{align}
    We denote 
    \begin{align*}
        &v_1=2^\epsilon\begin{pmatrix}-\frac{p+q(4m-1)}{2}\\ nqb- mq\end{pmatrix}&&v_{-1}=2^\epsilon\begin{pmatrix}\frac{q(4m-1)-p}{2}\\ mq+nqb\end{pmatrix}
    \end{align*} 
    First we show that $v_{\pm1}\in \bbZ^2$: If $\epsilon=1$, the above is evidently in $\bbZ^2$. If $\epsilon=0$, then we note that 
    \begin{align*}
        p\pm q=\tilde{p}\pm \tilde{q} -b(\tilde{p}+\tilde{q})\pm a(\tilde{p} +\tilde{q})=(1-b\pm a)\tilde{p}-\tilde{q}(b\mp 1\mp a)
    \end{align*}
    noting that in this case, both expressions above are divisible by $2$, so we get that 
    \begin{align*}
        &p+q(4m-1)\equiv p-q\equiv 0 \mod 2&&p-q(4m-1)\equiv p-q\equiv 0 \mod 2
    \end{align*}
    as needed. \par
    Then direct computation will verify that $N\calA_1v_1+N\calA_{-1}v_{-1}=\varrho_{-1}$, as needed.\par
    In the case where $N\kappa^\theta =\nu $, we need to show that 
    \begin{align*}
        \begin{pmatrix}0\\-1\end{pmatrix}\in N\calA_1\bbZ^2+N\calA_{-1}\bbZ^2
    \end{align*}
    We note that Equation (\ref{abBezNotilde}) still holds, with the same $p,q$ which are defined as above, then we consider 
    \begin{align*}
        &v_1=2^\epsilon\begin{pmatrix}\frac{q-p}{2}\\ -p\end{pmatrix}&&v_{-1}=2^\epsilon\begin{pmatrix}-\frac{p+q}{2}\\-p\end{pmatrix}
    \end{align*} 
    We have that $v_{\pm1}\in \bbZ$, as above, and again direct computation will verify that $N\calA_1v_1+N\calA_{-1}v_{-1}=\varrho_{-1}$.
\end{proof}
An immediate consequence of this is the following proposition that allows us to understand condition (\ref{ConditionForV}) better: 
\begin{restatable}{proposition}{DecomposeRhoProp}
\label{DecomposeRho}
    There are $v_{\pm1}\in\bbZ^2 $ such that 
    \begin{align*}
        \kappa^{\theta}\varrho_{-1}=R_{\theta} \kappa v_1+R_{-\theta} \kappa v_1
    \end{align*}
\end{restatable}
\begin{proof}
    By Proposition \ref{DescribeRho} we have that there is some $v_{\pm1}\in \bbZ^2$ such that  $N\calA_1v_1+N\calA_{-1}v_{-1}=\varrho_{-1}$.
    Recalling that $\calA_t=(N\kappa^\theta)^{-1} R_{\theta}^t \kappa$, for $t\in \{\pm1\}$ we can apply $N\kappa^\theta$ to both sides to get 
    \begin{align*}
        &N\kappa^\theta \varrho_{-1}=NR_{\theta }\kappa v_1+NR_{-\theta }\kappa v_{-1}\implies \kappa^\theta \varrho_{-1}=R_{\theta }\kappa v_1+R_{-\theta }\kappa v_{-1}
    \end{align*}
    as claimed.
\end{proof}
With this, we get the following result:
\begin{restatable}{lemma}{WDescriptionFirst}
\label{WDescription}%
    Let $H^\theta_{x}=-\Delta +\lambda W^\theta_{0,x}$, for $x\in \{AA,AB\}$, and $W_{0,x}^\theta $ defined in (\ref{AddPotAA}) or (\ref{AddPotAB}) for $\lambda\in \bbR$ and twisted bilayer potential with respect to honeycomb potentials $V$ and $Q$ (in AA stacking), and angle $\theta \in \calC\cap (0,\frac{\pi}{6})$.  Then we have that for any $\vec{m}\in \calS$, then we have for some $\ell\in \bbZ_3$
    \begin{align*}
        (\hat{W}_{0,x}^\theta)_{\vec{m}-\varrho_\ell}=0
    \end{align*}
    Furthermore, we have that 
    \begin{align*}
        (\hat{W}_{0,x}^\theta)_{\vec{m}} (\hat{W}_{0,x}^\theta)_{\vec{m}-\varrho_{-1}}\neq 0 \implies \exists t\in \{\pm1\},\vec{m}=N\calA_t  v_{t}+N^2\bbZ^2
    \end{align*}
    where $v_t$ are as in Proposition \ref{DecomposeRho}.
\end{restatable}
\begin{proof}[Proof of Lemma \ref{WDescription}]
    Let $W_{0,x}^\theta$ be as above, and assume that $\vec{m},\vec{m}-\varrho_1,\vec{m}-\varrho_{-1}\in \supp (\hat{W}^\theta_{0,x})$. Since 
    \begin{align*}
        \supp \hat{W}_{0,x}^\theta\subset N(\calA_1\bbZ^2\cup \calA_{-1}\bbZ^2)
    \end{align*}
    Then, by the pigeonhole principle, we have that two of the three vectors are in either $N\calA_1\bbZ^2$  or $N\calA_1^{-1}\bbZ^2$.  In other words,  we have that there are some
    $\ell,\ell'\in \bbZ_3$ and $t\in \{\pm1\}$ such that 
    \begin{align*}
        \vec{m}-\varrho_\ell,\vec{m}-\varrho_{\ell'}\in N\calA_t\bbZ^2\implies \varrho_{\ell'}-\varrho_{\ell}\in N\calA_t\bbZ^2
    \end{align*}
    But direct computation show that  $\varrho_{\pm1}, \varrho_{\pm1}-\varrho_{\mp1}\not \in \supp (\hat{W}^\theta_{0,x})$:
    \begin{align*}
        &\frac{1}{N}\calA_{\pm1}^{-1}\varrho_1 =\frac{1}{N^22^{\epsilon} }\begin{cases}
            \begin{pmatrix}\mp 2b\\ a\mp b\end{pmatrix}, &N\kappa^\theta=\kappa\\
            \frac{1}{4\pi}\begin{pmatrix}-2a\\-a\mp 3b\end{pmatrix}&N\kappa^\theta=\nu
        \end{cases}\not \in \bbZ^2\\
        &\frac{1}{N}\calA_{\pm1}^{-1} \varrho_{-1}=\frac{1}{N^22^{\epsilon}}\begin{cases}
            \begin{pmatrix}-a\mp b\\ \mp2 b\end{pmatrix},&N\kappa^\theta=\kappa\\
            \frac{1}{4\pi}\begin{pmatrix}-a\pm 3b\\ -2a\end{pmatrix},&N\kappa^\theta=\nu
        \end{cases}\not \in \bbZ^2\\
        &\frac{1}{N}\calA_{\pm1}^{-1}(\varrho_1-\varrho_{-1})=\frac{1}{N^22^{\epsilon}}\begin{cases}
            \begin{pmatrix}a\mp b\\ a\pm b\end{pmatrix},&N\kappa^\theta=\kappa\\
            \frac{1}{4\pi}\begin{pmatrix}-a\mp 3b\\ a\mp 3b\end{pmatrix},&N\kappa^\theta=\nu
        \end{cases}\not \in \bbZ^2
    \end{align*} 
    So we conclude that at least one of $\vec{m},\vec{m}-\varrho_1,\vec{m}-\varrho_{-1}$ are not $\supp(\hat{W}^\theta_{0,x})$ - as claimed. \par
    If we know that $\vec{m},\vec{m}-\varrho_{-1}\in \supp \hat{W}^\theta_{0,x}$, by the above we get that there is $t\in \{\pm1\}$ such that 
    \begin{align*}
        &\vec{m}\in N\calA_t \bbZ^2,&\vec{m}-\varrho_{-1}\in N\calA_{-t} \bbZ^2
    \end{align*}
    writing $\varrho_{-1}=N\calA_{1} v_1+N\calA_{-1} v_{-1}$ then we get that 
    \begin{align*}
        \vec{m}-N\calA_{t}v_t\in N\calA_{-t} \bbZ^2
    \end{align*}
    but since $\vec{m}\in N\calA_t \bbZ^2$ we may conclude that 
    \begin{align*}
        &\vec{m}-N\calA_{t}v_t\in N\calA_t \bbZ^2\cap N\calA_{-t} \bbZ^2\\
        &\kappa^\theta(\vec{m}-N\calA_{t}v_t)\in R_{\theta}\Lambda^*\cap R_{-\theta}\Lambda^*=N\begin{cases}
            \Lambda^* , & a\nmid 3\\
            \Lambda , & a\mid 3
        \end{cases}=N^2(\Lambda^\theta)^*
    \end{align*}
    The second-to-last equality comes from the proof of Lemma \ref{DiscrptionofLambda} when applied to $\Lambda^*$. Applying $(\kappa ^{\theta})^{-1}$ to this inclusion shows that $\vec{m}-N\calA_{t}v_t\in =N^2\bbZ^2$, as needed.
\end{proof}
Now we get as an immediate consequence Theorem \ref{DiracPointForTwisted}
\begin{proof}[ Proof of Theorem \ref{DiracPointForTwisted}]
    By \cite{fefferman2012honeycomb}, if we have that $\hat{W}^\theta_{\varrho_{1}}\neq0 $, we get the wanted result. In the other case, Theorem \ref{DiracPoint} holds for $W^\theta$. Thus completing the proof. 
\end{proof}
\subsection{Flattening of the Dirac cones for weak potential}
Finally, we prove our statement about the flattening of the cone for small potentials and angles close to incommensurate angles, Theorem \ref{Vanishing}. It is important to recall that this Theorem holds for \emph{all} twisted potentials, not only for potentials of the type of (\ref{AddPotAA}) or (\ref{AddPotAB}):
\begin{proof}[Proof of Theorem \ref{Vanishing}]
   Equation (\ref{asymptotics}), in the context of twisted potential, has the form of, for some $C>0$
   \begin{align*}
       |v_d(\lambda)|^2\leq &C(|K^\theta_*|^2+\lambda  \|W^\theta\|_\infty +\lambda^2 \|\nabla W^\theta\|_{\infty}^2\sum_{\vec{m}\in \calS\setminus \{\vec{0}\}} \frac{1}{|K^\theta_*+\kappa^\theta\vec{m}|^{4}})+O(\lambda^3\|W^\theta\|^3)
   \end{align*}
   We note that for any $\vec{m}\in \calS\setminus\{0\}$ we have some constant $c>0$ such that 
   \begin{align*}
       |K^\theta_*+\kappa^\theta\vec{m}|>c|k_1^\theta|
   \end{align*}
   Using the fact that the sum above can be treated as a Riemann sum, and bounding the resulting integral, we have that for some constant $C>0$, 
   \begin{align*}
       \sum_{\vec{m}\in \calS\setminus \{\vec{0}\}} \frac{1}{|K^\theta_*+\kappa^\theta\vec{m}|^{4}}\leq\frac{C}{c|k_1^\theta|^2}
   \end{align*}
   We note that using the boundedness of $G$, we can write
   \begin{align*}
        &\|W^\theta\|_{\infty }\leq C_g\|V\|_{\infty }^{\gamma}\|Q\|_{\infty }^{\gamma}, &\|\nabla W^\theta\|_{\infty }\leq C_{g'} \|\nabla V\|_{\infty }^{\gamma'}\|\nabla Q\|_{\infty }^{\gamma'}
    \end{align*}
    So we get that for some constant $C>0$, we have 
    \begin{align*}
       &|v_d(\lambda)|^2\leq C(\frac{1}{N^2}|NK^\theta_*|^2+\lambda  \|V\|_\infty^{\gamma} \|Q\|_\infty^{\gamma} +\lambda^2 \|\nabla V\|_{\infty}^{2\gamma'}\|\nabla Q\|_{\infty}^{2\gamma'}2N^2|Nk_1^\theta|^{-2})+O(\lambda^3\|V\|_\infty^{3\gamma}\|Q\|_\infty^{3\gamma})
   \end{align*}
   Recalling that $N\kappa^\theta \in \{\kappa,\nu\}$, and so is independent of $N$ in terms of sizes (up to a factor of $\frac{4\pi}{3}$), so we have that
   \begin{align*}
      & |NK^\theta_*|^2=O(1)=|Nk_1^\theta|^2
   \end{align*}
   as $N\rightarrow \infty$. Thus, we get that for some $C>0$ depending only on $V$ and $Q$ such that 
   \begin{align*}
       |v_d(\lambda)|^2\leq C(\frac{1}{N^2} +\lambda +\lambda^2 N^2)+O(\lambda^3\|V\|_\infty ^{3\gamma}\|Q\|_\infty ^{3\gamma})
   \end{align*}
   In particular, we get that if  $ |\lambda|<\frac{\delta}{N^2}$ for some $\delta>0$, we have that 
   \begin{align*}
       \lambda +\lambda^2 N^{2}<\frac{(\delta+1)^2}{N^2}
   \end{align*}
   Since $\|V\|_\infty ^{3\gamma}\|Q\|_\infty ^{3\gamma}$ is independent of $N$, for such $\lambda$ there exists a constant $C(\delta, V, G) > 0$ such that:
   \begin{align*}
       |v_d(\lambda)|^2\leq \frac{C}{N^2}+O(\lambda^3)
   \end{align*}
    So, we may conclude that we have that 
   \begin{align*}
       |\lambda|<\frac{\delta}{N^2}\implies |v_d| \leq \frac{C(\delta,V,G)}{N}+O(N^{-3})
   \end{align*}
   for some $\delta,C(\delta,V,G)>0$ as claimed.
\end{proof}

\section{Examples}\label{Example}
In this section, we construct a set of examples of potentials of the type of $W^\theta_{0, AA}$ for which the above theorems hold. We recall the proposition:
\VConstruction*
\begin{proof}
    First, the fact that $V$ defined above is a honeycomb potential is immediate as it is periodic with respect to $\Lambda$, real and even. The relation $\hat{V}_{\vec{m}}=\hat{V}_{B^{\pm1}\vec{m}}$ implies spatial symmetry with respect to $R$. Finally, the exponential decay of $a_{\vec{m}}$ ensures that $V\in C^\infty$, as required. \par
    Now, we recall that, by Lemma \ref{WDescription}
    \begin{align*}
        \hat{W}^\theta_{\vec{m}} \hat{W}^\theta_{\vec{m}-\varrho_{-1}}\neq 0 \implies \exists t\in \{\pm1\},\vec{m} \in N\calA_{t}v_t+N^2\bbZ^2
    \end{align*}
    And so we have that 
    \begin{align*}
        \sum_{\vec{m}\in \calS\setminus \{\vec{0}\}}\frac{\hat{W}^\theta_{\vec{m}}\hat{W}^\theta_{\vec{m}-\varrho_{-1}}}{|K_*^\theta(\vec{m})|^2-|K_*^\theta|^2}&=\sum_{t\in \{\pm1\}, u\in \bbZ^2}\frac{\hat{W}^\theta_{N\calA_{t}v_t+N^2u}\hat{W}^\theta_{N\calA_{-t}v_{-t}+N^2u}}{|K_*^\theta(N\calA_{t}v_t+N^2u)|^2-|K_*^\theta|^2}\\
        &=\sum_{t\in \{\pm1\}, u\in \bbZ^2}\frac{\hat{V}_{v_t+N\calA_{-t}u}\hat{V}_{v_{-t}+N\calA_{t}u}}{|K_*^\theta(N\calA_{t}v_t+N^2u)|^2-|K_*^\theta|^2}
    \end{align*}
    where we used the fact that $N\calA_{t}v_t\in N\calA_t\bbZ^2\setminus( N\calA_{-t}\bbZ^2)$, and the identification between the Fourier coefficients implied by Equation (\ref{FourierId}). If $V$ is defined with a positive sign, each summand in the above sum is strictly positive; if defined with a negative sign, each summand is strictly negative. In either case, we obtain:
    \begin{align*}
        \sum_{\vec{m}\in \calS\setminus \{\vec{0}\}}\frac{\hat{W}^\theta_{\vec{m}}\hat{W}^\theta_{\vec{m}-\varrho_{-1}}}{|K_*^\theta(\vec{m})|^2-|K_*^\theta|^2}\neq0
    \end{align*}
    and we conclude that condition (\ref{ConditionForW}) holds. 
\end{proof}

\bibliographystyle{amsplain}
\bibliography{bib}
\newpage
\appendix
\section{Explicit Lattice Computation}\label{ExpLatCom}
In this appendix, we explicitly solve the congruence system that arises in the proof of Lemma \ref{DiscrptionofLambda}.
\begin{proposition}
    Let $a, b$ be coprime integers with $a > 0$, and let $\calI = \begin{pmatrix} 1 & 2 \\ -2 & -1 \end{pmatrix}$. Suppose $u = \begin{pmatrix} u_1 \\ u_2 \end{pmatrix} \in \bbZ^2$ satisfies the congruence system:
    \begin{align}
        0 &\equiv (a+b)u \pmod 2 \label{Mod2}\\
        0 &\equiv \calI u \pmod a \label{Moda}
    \end{align}
    Furthermore, suppose $u$ satisfies the boundary condition that $\frac{b}{a}\calI u \in \bbZ^2$. Let $\epsilon=\begin{cases} 1,&2\nmid ab\\ 0,& 2\mid ab \end{cases} $ and $\rho=\begin{cases}1,& 3\mid a\\0,& 3\nmid a\end{cases} $. 
    Then $u$ must take the form:
    \begin{align*}
        u = 2^{1-\epsilon} a \left(p+\frac{\rho r}{3^{\rho}}\begin{pmatrix}1\\1\end{pmatrix}\right)
    \end{align*}
    for some $p \in \bbZ^2$ and $r \in \bbZ_3$.
\end{proposition}
\begin{proof}
    Equation (\ref{Mod2}) implies that $2^{\epsilon-1}u\in \bbZ^2$, for $\epsilon$ as above. Equation (\ref{Moda}) implies
   \begin{align*}
       0\equiv 2^{\epsilon-1}\begin{pmatrix} u_1+2u_2\\ -2u_1-u_2\end{pmatrix}\mod a\implies a\mid 3\cdot 2^{\epsilon-1}(u_1+u_2)
   \end{align*}
    Writing $a=3^\rho c$, for $\rho$ as above,  we have that $c\mid 3^{1-\rho}2^{\epsilon-1}(u_1+u_2), 2^{\epsilon-1}(2u_1+u_2)$ which implies that $c\mid 3^{1-\rho}2^{\epsilon-1} u_1,3^{1-\rho}2^{\epsilon-1} u_2$. Since $c\nmid 3^{1-\rho}$, we get that $c\mid 2^{\epsilon-1}u_1,2^{\epsilon-1}u_2$, so $c^{-1}2^{\epsilon-1}u \in \bbZ^2$. Then we have that 
    \begin{align*}
        \bbZ^2\ni\frac{b}{a}\calI u=\frac{2^{1-\epsilon} b}{ 3^{\rho}}\calI (2^{\epsilon-1}c^{-1}u)
    \end{align*}
     If $\rho=0$, it is evident that $\frac{b}{a}\calI u\in \bbZ^2$. If $\rho =1$, we need in particular that  
    \begin{align*}
        \frac{1}{3}\calI 2^{\epsilon-1}c^{-1}u \in \bbZ^2
    \end{align*}
    as $3\nmid b$. Thus, we need that $3\mid u_1+2u_2,2u_1+u_2$ which implies that $u_1\equiv u_2\mod 3$, so we can write,
    \begin{align*}
        &2^{\epsilon-1}c^{-1}u=3p+r\begin{pmatrix}1\\1\end{pmatrix}\implies u=2^{1-\epsilon} c (3p+r\begin{pmatrix}1\\1\end{pmatrix})=2^{1-\epsilon} a(p+\frac{r}{3}\begin{pmatrix}1\\1\end{pmatrix})
    \end{align*}
    for $r\in \bbZ_3$. Combining both cases, we get 
    \begin{align*}
        u=2^{1-\epsilon} a (p+\frac{\rho r}{3^{\rho}}\begin{pmatrix}1\\1\end{pmatrix})
    \end{align*}
    as claimed.
\end{proof}
\newpage
\section{Notation}\label{Notations}
\begin{itemize}
    \item We denote by $\braket{\cdot.\cdot}$ the Euclidean inner product on vectors in $\bbR^2$ or $\bbZ^2$, and the size of these vector is denoted by $|\cdot|$. 
    \item Throughout the paper, $\tilde{\Lambda}$ denotes a generic honeycomb lattice (without explicit reference to its base vectors), $\Lambda$ denotes a honeycomb lattice with base vectors defined by 
    \begin{align*}
        v_1=\begin{pmatrix}\frac{\sqrt{3}}{2}\\\frac{1}{2}\end{pmatrix}, v_2=\begin{pmatrix}\frac{\sqrt{3}}{2}\\-\frac{1}{2}\end{pmatrix}
    \end{align*}
    $\Lambda^*$ denotes the dual to $\Lambda$, and $\Lambda^\theta$ denote the lattice with respect to which $W^\theta$ is periodic.
    \item $\tilde{\nu}$ is the base matrix of $\tilde{\Lambda}$, and $\tilde{\kappa}$ is the base matrix of $\tilde{\Lambda}^*$.
    \item $R_\theta$ denotes the rotation matrix by $\theta$, and we denote $R=R_{\frac{2\pi}{3}}$, their corresponding operators be denoted by $\calR_{\theta}$ and $\calR$ respectively.
    \item We denote by  $\calT_{\vec{a}}$  the translation operator by the vector $\vec{a}$, for any $\vec{a}\in \bbR^2$. 
    \item We denote by $\tilde{\Omega}=\tilde{\nu}[0,1]^2$- the unit cell,  by $\tilde{\calB}=\{k\in \bbR^2\mid \forall a\in \tilde{\Lambda}^*, |k|\leq |k-a|\}$ the Brillouin zone, and the points of high symmetry by 
    \begin{align*}
        \tilde{\bbP}&=\{\vec{k}\in \tilde{\calB}\mid (R-\id)\vec{k}\in \tilde{\kappa} \bbZ^2\}\\
        &=\{\tilde{K},R\tilde{K},R^2\tilde{K}\}\bigsqcup \{\tilde{K}',R\tilde{K}',R^2\tilde{K}'\}\bigsqcup\{0\}
    \end{align*}
    \item  We distinguish one of the points of high symmetry 
    \begin{align*}
        K_0=\frac{1}{3}\nu \begin{pmatrix}1\\1\end{pmatrix}=\begin{pmatrix}\frac{1}{\sqrt{3}}\\0\end{pmatrix},
    \end{align*}
    and its dual-point 
    \begin{align*}
        K_0^*=\frac{1}{3}\kappa \begin{pmatrix}1\\-1\end{pmatrix}=\frac{4\pi}{\sqrt{3}}\begin{pmatrix}0\\\sqrt{3}\end{pmatrix}
    \end{align*}
    \item Throughout the paper, $V$ and $Q$ denotes honeycomb potentials used to define the twisted bilayer potential 
    \begin{align*}
        &W^\theta_{AA}= G(\calR_{\theta}V,\calR_{-\theta}Q)\\
        &W_{AB}^\theta=G^*(\frac{1}{3}\sum_{j=-1}^1\calT_{-\frac{1}{2}R^jK_0}\calR_\theta V,\frac{1}{3}\sum_{j=-1}^1\calT_{\frac{1}{2}R^jK_0}\calR_{-\theta} V)
    \end{align*}
    for $G$ and $G^*$ admissible interaction operators, respectively, and we use $U$ to denote a generic honeycomb potential, which is periodic with respect to $\tilde{\Lambda}$.
    \item We denote $\tau=-\frac{1}{2}+\frac{\sqrt{3}}{2}i=e^{-\frac{2\pi}{3}i}$ the cubic root of unity. 
    \item $\calC$ is the set of commensurate angles.
    \item We consider the following spaces, for $k\in \tilde{\calB}$
    \begin{align*}
        &L^2_{k}(\tilde{\Omega})=\{f\in L^2(\tilde{\Omega} )\mid \forall a\in \tilde{\Lambda}, f(x+a)=e^{-i\braket{k,a}}f(x)\}\\
        &L^2_{k,\sigma}(\tilde{\Omega})=\{f\in L^2_{k}(\tilde{\Omega})\mid \calR f=\sigma f\}    
    \end{align*}
    for $\sigma \in \{1,\tau,\bar{\tau}\}$. 
    \item For $f\in L^2_{0}=L^2_{per}$, we have the following Fourier representation:
    \begin{align*}
        &\hat{f}_{\vec{m}}=\frac{1}{|\Omega|}\int\limits_{\Omega}e^{-i\braket{\kappa \vec{m},y}}f(y)\, dy\\
        &f(y)=\sum_{\vec{m}\in \bbZ^2}\hat{f}_{\vec{m}} e^{i\braket{\kappa \vec{m},y}}
    \end{align*}
    \item We denote by $(\cdot,\cdot)$ the inner product on $L^2_{k}(\Omega)$ spaces, and norms are denoted by $\|\cdot\|$.
    \item We denote the following 
    \begin{align*}
        &B=\tilde{\kappa}^{-1}R\tilde{\kappa}\\
        &\varrho_1=\tilde{\kappa}^{-1}(R-\id)\tilde{K}_*\\
        &\varrho_{-1}=\tilde{\kappa}^{-1}(R^{-1}-\id)\tilde{K}_*\\
        &\varrho_0=0
    \end{align*}
    \item We define the equivalence $\approx$ that identifies the orbit of $\vec{m}$ under $B^j\vec{m}+\varrho_j, j\in \bbZ_3$, and we denote $\calS=\bbZ^2/\approx$. 
    \item For $\theta \in \calC\cap (0,\frac{\pi}{3})$, we have $\tan(\theta)=\frac{\sqrt{3}b}{a}$ for some co-prime $a,b\in \bbZ$, such that $0<b<\frac{a}{b}$, and we denote 
    \begin{align*}
        &\alpha =\begin{cases}
            8\pi, & 3\mid a \text{ and } 2\nmid ab\\
            2, & 3\nmid a \text{ and } 2\nmid ab\\
            4\pi , & 3\mid a \text{ and } 2\mid ab\\
           1, & 3\nmid a \text{ and } 2\mid ab
        \end{cases}\\
        &N=\frac{1}{\alpha}\sqrt{a^2+3b^2}
    \end{align*}
    \item We denote 
        \begin{align*}
            &\calA_1=( N\kappa^\theta)^{-1} R_{\theta}\kappa\\
            &\calA_{-1}=( N\kappa^\theta)^{-1} R_{-\theta}\kappa
        \end{align*} 
\end{itemize}
\end{document}